\documentclass[a4paper,12pt,leqno]{article}
\usepackage[latin1]{inputenc}
\usepackage{amssymb}
\usepackage{enumerate}
\usepackage{amsmath}
\usepackage{amsthm}
\usepackage{latexsym}
\usepackage{graphicx}
\usepackage{bm}
\usepackage{overpic}
\usepackage[normalem]{ulem}

\usepackage{exscale}
\usepackage{amsfonts}
\usepackage[usenames]{color} 
\definecolor{darkgreen}{rgb}{.2,.8,0}

 \def\ua{\uparrow}
 \def\da{\downarrow}

 \def\wh{\widehat}
 \def\wt{\widetilde}

\def\bbar{\overline}
\def\bR{\mathbb{R}}

\def\cM{\mathcal M}
\def\cN{\mathcal N}
\def\cR{\mathcal R}
\def\bR{\mathbb R}
\def\cX{\mathcal X}
\def\cF{\mathcal F}
\def\bP{\mathbb P}
\def\bE{\mathbb E}
\def\bN{\mathbb N}

\def\eps{\varepsilon}

\newcommand{\avatr}{{\rm AV@R}}
\newcommand{\vatr}{{\rm V@R}}
\newcommand{\N}{\mathbb{N}}
\newcommand{\R}{\mathbb{R}}

\newcommand{\eins}{\mathbbm{1}}
\DeclareMathOperator{\esssup}{ess\,sup}

\newcommand{\pr}{\mathbb{P}}
\newcommand{\ex}{\mathbb{E}}

\numberwithin{equation}{section}
\newtheorem{theorem}{Theorem}[section]
\newtheorem{proposition}[theorem]{Proposition}
\newtheorem{lemma}[theorem]{Lemma}
\newtheorem{corollary}[theorem]{Corollary}

\theoremstyle{definition}
\newtheorem{definition}[theorem]{Definition}
\newtheorem{example}[theorem]{Example}

\newtheorem{remark}[theorem]{Remark}

\def\eins{{\mathchoice {1\mskip-4mu\mathrm l}
{1\mskip-4mu\mathrm l}{1\mskip-4.5mu\mathrm l}
{1\mskip-5mu\mathrm l}}}

\renewcommand{\baselinestretch}{1.0}\normalsize

\begin{document}
\title{\LARGE\bf Comparative and qualitative robustness for\\ law-invariant risk measures}
\author{
Volker Kr\"atschmer\footnote{Faculty of Mathematics, University of Duisburg-Essen, {\tt volker.kraetschmer@uni-due.de}} {\setcounter{footnote}{2}}  \qquad
Alexander Schied\footnote{Department of Mathematics, University of Mannheim, {\tt schied@uni-mannheim.de}} {\setcounter{footnote}{6}}  \qquad
Henryk Z\"ahle\footnote{Department of Mathematics,  Saarland University, {\tt zaehle@math.uni-sb.de}}
}

\date{~
}

\maketitle

\begin{abstract}
When estimating the risk of a P\&L from historical data or Monte Carlo simulation, the robustness of the estimate is important. We argue here that Hampel's classical notion of qualitative robustness is not suitable for risk measurement and we propose and analyze a refined notion of robustness that applies to tail-dependent law-invariant convex risk measures on Orlicz spaces. This concept of robustness captures the tradeoff between robustness and sensitivity and can be quantified by  an index of qualitative robustness. By means of this index, we can compare various risk measures, such as distortion risk measures, in regard to their degree of robustness. Our analysis also yields results that are of independent interest such as continuity properties and consistency of estimators for risk measures, or a Skorohod representation theorem for $\psi$-weak convergence.

\medskip

\parindent0pt
{\bf JEL Classification} D81
\end{abstract}


\section{Introduction}

Let $X$ denote the  P\&L of a financial position.
When assessing the risk of $X$   in terms of a monetary risk measure $\rho$ it is   common  to estimate  $\rho(X)$ by means of a Monte Carlo procedure or from a sequence of historical data. When $\rho$ is a law-invariant risk measure, a natural estimate for $\rho(X)$ is given by $\cR_\rho(\wh m)$, where $\wh m$ is the  empirical distribution of the data and $\cR_\rho$ is the functional defined  by
$$\cR_\rho(\mu)=\rho(X)\qquad\text{if $X$ has law $\mu$;}
$$ see, e.g.,   \cite{AcerbiTasche,BelomestnyKraetschmer,BeutnerZaehle2010,Contetal,Weber}.
In this context,  it was pointed out by Cont et al. \cite{Contetal} that  it is important to consider the \emph{robustness} of the  risk functional $\cR_\rho$.
 Informally, robustness refers here to a certain insensitivity of the sampling distribution with respect to deviations of $\widehat{m}$ from the \lq true\rq,  theoretical distribution. It  will especially yield a  stable behavior of the estimates when the estimation process is repeated periodically. Such a stable behavior  is particularly desirable when $\rho(X)$ serves in allocating the economic capital required from a large position $X$, since altering the capital allocation may be costly.

On a  mathematical level, Cont et al. \cite{Contetal} use Hampel's \cite{Hampel} classical concept of qualitative robustness, which, according  to Hampel's theorem, is essentially equivalent to the weak continuity of   $\cR_\rho$.  Consequently, it was pointed out in  \cite{Contetal} that no risk functional  $\cR_\rho$ that arises from a law-invariant coherent risk measure $\rho$
can satisfy   Hampel's requirement of qualitative robustness, not even if $\rho$ is simply the ordinary expectation of the loss.  The results in  \cite{Contetal} therefore seem to weigh heavily in favor of Value at Risk, since Value at Risk does essentially satisfy  Hampel's  notion of  robustness.

Our goal in this article is to point out that in risk measurement  the use of Hampel's classical concept of qualitative robustness may be problematic and to propose and analyze an alternative concept  based on \cite{Kraetschmer et al 2012}.
Let us start  by discussing two major drawbacks of Hampel's robustness in risk measurement.

First, two P\&Ls may have laws that are close with respect to the weak topology but still have completely different tail behavior. Qualitative robustness of $\cR_\rho$ therefore requires that $\rho$ is essentially insensitive to the tail behavior of a P\&L. In the recent years of financial crisis, it has become apparent, though, that a faulty assessment of  tail behavior
can lead to a dramatic underestimation of the corresponding risk.

Second, Hampel's robustness concept  creates a sharp  division of the class of law-invariant monetary risk measures into those for which $\cR_\rho$ is \lq robust\rq\ and those for which $\cR_\rho$ is \lq not robust\rq. The first class contains risk measures such as Value at Risk that are insensitive with respect to tail behavior of P\&Ls whereas the second class contains the ordinary expectation and all law-invariant coherent risk measures \cite{Contetal}. But, as we will see, the distinction between \lq robust\rq\  and \lq non-robust\rq\ risk measures is artificial because there is actually a full continuum of possible degrees of robustness beyond the classical concept. So labeling a risk measure as \lq robust\rq\  or \lq non-robust\rq\ may give a false impression.

In this article, we will analyze the robustness properties of law-invariant convex risk measures based on the refined notion of qualitative robustness that was proposed in \cite{Kraetschmer et al 2012} and is further developed here in  Section \ref{qualitative robustness}. Instead of a sharp division into \lq robust\rq\  and \lq non-robust\rq\ risk measures, this notion allows us to assign a degree of robustness to most risk measures and to compare different risk measures in regard to their degrees of robustness.  We thereby capture the natural tradeoff between robustness and tail sensitivity in risk measurement.
The degree of robustness can be expressed numerically by the \emph{index of qualitative robustness} proposed in \cite{Kraetschmer et al 2012}.  This index takes values in $[0,\infty]$, with the respective extremes $+\infty$ and $0$ corresponding to Hampel's robustness and to full tail sensitivity. Some of our main results will  show that a greater index of qualitative robustness implies greater robustness in a sense that is mathematically precise.  We will also show how our index can be computed for  distortion risk measures such as MINMAXVAR or Average Value at Risk (which is also called Expected Shortfall, Conditional Value at Risk, or TailVaR). We emphasize that the key to our refined notion of robustness lies in specifying a metric on a suitable space  of probability measures  for which the statistical functional associated with our risk measure becomes continuous. It follows from the results in \cite{Contetal} that such a metric must generate a topology that is \emph{finer} than the usual topology of weak convergence.


Since we are  interested in the way in which  convex risk measures depend on the tail of a P\&L,   it is  not sufficient to consider only bounded P\&Ls.  We therefore build on  the analysis of  Cheridito and Li \cite{CheriditoLi2009}, who observed that Orlicz spaces or Orlicz hearts are  appropriate domains for convex risk measures when P\&Ls are unbounded. We are particularly interested in the continuity properties of $\rho$ and its corresponding risk functional  $\cR_\rho$, and we find that an important role is played by the so-called $\Delta_2$-condition of the underlying Orlicz space.

Our analysis also yields some results that are of independent interest. For instance, we obtain a Skorohod representation theorem that links $\psi$-weak convergence of probability measures to norm convergence of random variables in Orlicz space.

Our article is organized as follows. In Section \ref{Consistency section} we prove the consistency of the estimator  $\cR_\rho(\wh m)$ for law-invariant convex risk measures and general stationary and ergodic data. In Section \ref{psi-weak topology and qualitative robustness} we analyze the continuity properties of $\cR_\rho$ and show that $\cR_\rho$ basically inherits the continuity of the original risk measures if and only if  the Orlicz space supporting $\rho$ satisfies the $\Delta_2$-condition. In Section \ref{qualitative robustness} we present our main results on the comparative and qualitative robustness of law-invariant convex risk measures $\rho$. In particular, we introduce our refined notion of robustness and the index of qualitative robustness, and we show that these notions are well-defined whenever the Orlicz space supporting $\rho$ satisfies the $\Delta_2$-condition. In Section \ref{distortion index} we show that our results can be easily applied to distortion risk measures.

In Section \ref{Robust section} we continue and strengthen the  robustness analysis for general statistical functionals started in \cite{Kraetschmer et al 2012}. In particular, we state stronger versions of our Hampel-type theorem and its converse than those given in \cite{Kraetschmer et al 2012}.  In Section \ref{Skorohod Section} we state and prove our above-mentioned Skorohod representation result. Most other proofs can be found in Section \ref{Proof section} and the appendix.


\section{Statement of main results}\label{main results section}

\subsection{Setup}\label{Setup section}

Let $(\Omega,\cF,\bP)$ be an atomless probability space and denote by $L^0:=L^0(\Omega,\cF,\bP)$ the usual class of all finitely-valued random variables modulo the equivalence relation of $\bP$-a.s.\ identity. Let $\cX\subset L^0$ be a vector space containing the constants. An element $X$ of $\cX$ will be interpreted as the P\&L of a financial position.  We will say that a map $\rho:\cX\to\bR$ is  a \emph{convex risk measure} when the following conditions are satisfied:
\begin{description}
\item[\rm (i)] monotonicity: $\rho(X)\ge\rho(Y)$ for $X$, $Y\in\cX$ with $X\le Y$;
\item[\rm (ii)] convexity:  $\rho(\lambda X+(1-\lambda) Y)\le\lambda\rho(X)+(1-\lambda)\rho(Y)$ for all $X,Y\in\cX$ and $\lambda\in[0,1]$;
\item[\rm (iii)] cash additivity: $\rho(X+m)=\rho(X)-m$ for $X\in\cX$ and $m\in\bR$.
\end{description}

\begin{remark}\label{cash coercivity rem}
It was argued in \cite{ElKarouiRavanelli}  that the requirement of cash additivity should be relaxed when interest rates are stochastic or ambiguous, or when bonds are subject to possible default. As a matter of fact, in our results it is possible to replace axiom (iii) by the following weaker notion:
\begin{description}
\item[\rm (iii')] cash coercivity: $\rho(-m)\longrightarrow+\infty$ when  $m\in\bR$ tends to $+\infty$.
\end{description}
To keep our presentation simple, we have however stated our results within the standard framework of cash-additive convex risk measures.\end{remark}

\medskip

As discussed, e.g., in  \cite{AcerbiTasche,BelomestnyKraetschmer,BeutnerZaehle2010,Contetal} it is a common procedure to estimate the risk of a financial position by means of a Monte Carlo procedure or from a sequence of historical data. Such a procedure makes sense when $\rho$ is \emph{law-invariant}: $\rho(X)=\rho(\wt X)$ whenever $X$ and $\wt X$ have the same law under $\bP$.
Let us denote by $\cM(\cX):=\{\bP\circ X^{-1}\,:\,X\in\cX\}$ the class of all Borel probability measures on $\bR$ that arise as the distribution of some $X\in\cX$. Law invariance of a risk measure $\rho:\cX\to\bR$ is equivalent to the existence of a map $\cR_\rho:\cM(\cX)\to\bR$ such that
\begin{equation}\label{}
\rho(X)=\cR_\rho(\bP\circ X^{-1}),\qquad X\in\cX.
\end{equation}
This map $\cR_\rho$ will be called the \emph{risk functional} associated with $\rho$. It can  be used in a natural way to construct estimates for the risk $\rho(X)$ of $X\in\cX$. All one has to do is to  take an estimate $\wh\mu_n$ for the law $\mu=\bP\circ X^{-1}$ of $X$ and to plug this estimate into the functional $\cR_\rho$ to get the desired estimator:
\begin{eqnarray}\label{rho estimator - 1}
\wh\rho_n:=\cR_\rho(\wh\mu_n);
\end{eqnarray}
see, e.g., \cite{AcerbiTasche,BelomestnyKraetschmer,BeutnerZaehle2010,Contetal,Weber}.
For instance, $\wh\mu_n$ can be the empirical distribution $\frac1n\sum_{k=1}^n
\delta_{x_k}$ of a sequence $x_1,\dots, x_n$ of historical observations or  Monte Carlo simulations.

\begin{example}\label{mean example}In this very basic example we take $\cX=L^1$ and $\rho(X)=-\bE[\,X\,]$. Then $\rho$ is a law-invariant coherent risk measure and $\cR_\rho(\mu)$ is simply the negative mean of the measure $\mu$, i.e., $\cR_\rho(\mu)=-\int x\,\mu(dx)$.
\end{example}

There are  two natural questions that arise in this context. The first question refers to the \emph{consistency} of a sequence of estimates $\wh\rho_n$. That is, under which conditions do we have $\wh\rho_n\to\rho(X)$ as $n\ua\infty$? When we assume that the estimates $\wh\mu_n$ converge to $\mu$ in some suitable topology on the space of measures, then the consistency of $\wh\rho_n$ boils down to establishing the continuity of $\cR_\rho$ in that topology. We will thus also analyze the continuity properties of $\cR_\rho$.

Once consistency and continuity have been established, one can investigate the \emph{robustness}
of the estimate $\wh\rho_n$. Informally robustness refers to the stability of  $\wh\rho_n$ with respect to small perturbations of the law under which the data points $x_1,\dots, x_n$  are generated. The issue of robustness of this plug-in method for risk measures was first raised in \cite{Contetal}. Here we will address it in Section \ref{qualitative robustness}.

\bigskip

Before stating our results, we need to specify the setting in which we are going to work.
A common choice for $\cX$ is the space $L^\infty:=L^\infty(\Omega,\cF,\bP)$ of all bounded random variables. When dealing with possibly unbounded risks, however, the choice $\cX=L^\infty$ is not suitable. It was observed in \cite{BiaginiFrittelli,BiaginiFrittelli2,CheriditoLi2009} that Orlicz spaces or Orlicz hearts may be  appropriate choices for $\cX$ when risks may be unbounded.  Let us thus recall the basic notions of Orlicz spaces. Following \cite{CheriditoLi2009},  a Young function will be a left-continuous, nondecreasing convex function $\Psi:\bR_+\to[0,\infty]$ such that $0=\Psi(0)=\lim_{x\da0}\Psi(x)$ and $\lim_{x\ua\infty}\Psi(x)=\infty$. Such a function is continuous except possibly at a single point at which it jumps to $+\infty$.  The Orlicz space associated with $\Psi$ is
$$L^\Psi:=L^\Psi(\Omega,\cF,\bP)=\big\{X\in L^0\,:\,\bE[\,\Psi(c|X|)\,]<\infty\text{ for some $c>0$}\big\}.
$$
It is a Banach space when endowed with the Luxemburg norm,
 $$
    \|X\|_{\Psi} := \inf\left\{\lambda > 0\,:\,\bE [\,\Psi(|X|/\lambda)\,]\leq 1\right\}.
$$
We will frequently use the following equivalent property for convergence of a sequence $(X_n)_{n\in\bN_0}\subset L^\Psi$:
\begin{equation}\label{Proposition 2.1.10 in EdgarSucheston1992}
\|X_n-X_0\|_\Psi\longrightarrow0\quad\text{if and only if}\quad \bE\big[\,\Psi(k|X_n-X_0|)\,\big]\longrightarrow0\quad\text{for all $k>0$;}
\end{equation}
see Proposition 2.1.10 in \cite{EdgarSucheston1992}.
The Orlicz heart is defined as
$$H^\Psi:=H^\Psi(\Omega,\cF,\bP)=\big\{X\in L^0\,:\,\bE[\,\Psi(c|X|)\,]<\infty\text{ for all $c>0$}\big\}.
$$
When $\Psi$ takes the value $+\infty$, then $H^\Psi=\{0\}$ and $L^\Psi=L^\infty$. For this reason, we will mainly focus on the case in which $\Psi$ is finite. Then $L^\infty\subset H^\Psi\subset L^\Psi\subset L^1$, and these inclusions may all be strict. In fact for finite $\Psi$, the identity $H^\Psi=L^\Psi$ holds if and only if $\Psi$ satisfies the so-called
 $\Delta_2$-condition,
\begin{equation}\label{Delta2}
    \text{there are $C$, $x_0>0$ such that $\Psi(2x)\le C\Psi(x)$ for all $x\ge x_0$; }
\end{equation}
see \cite[Theorem 2.1.17 (b)]{EdgarSucheston1992}. This condition is clearly satisfied when specifically $\Psi(x)=x^p/p$ for some $p\in[1,\infty)$. In this case,   $H^\Psi=L^\Psi=L^p$ and $ \|Y\|_{\Psi}=p^{-1/p}\|Y\|_p$.


\begin{example}[Risk measure based on one-sided moments]\label{Risk measure based on one-sided moments}
The risk measure based on one-sided moments is defined as
\begin{equation}\label{Risk measure based on one-sided moments - eq}
    \rho(X):=-\ex[X] + a\,\ex[((X-\ex[X])^{-})^p]^{1/p},
\end{equation}
where $p\in [1,\infty)$ and $a\in[0,1]$ are constants; see also \cite{Delbaen2002}. It is well-defined and finite on $L^p$, law-invariant, and it is easily seen that it satisfies the axioms of a convex risk measure. \hfill$\diamondsuit$
\end{example}
\medskip

When the $\Delta_2$-condition \eqref{Delta2} is not satisfied, then the Orlicz heart rather than the Orlicz space $L^\Psi$ is the natural domain for a convex risk measure as is illustrated by the following examples.

\medskip

\begin{example}[Entropic risk measure]\label{entropic rm ex} The entropic risk measure is defined as
\begin{equation}\label{entropic rm eq}
\rho (X):=\frac1\beta\log\bE[\,e^{-\beta X}\,],
\end{equation}
where $\beta$ is a positive constant; see \cite[Example 12]{FoellmerSchied2002}. It is well-defined  and finite on the Orlicz heart $H^\Psi$ for the Young function $\Psi(x)=e^x-1$, but it is clearly not finite on the entire Orlicz space $L^\Psi$. Clearly, $\Psi$ does not satisfy the  $\Delta_2$-condition \eqref{Delta2}. The associated risk functional $\cR_\rho:\cM(H^\Psi)\to\bR$ is given by
$\cR_\rho(\mu)=\frac1\beta\log\int e^{-\beta x}\,\mu(dx)$.\hfill$\diamondsuit$
\end{example}
\medskip

\begin{example}[Utility-based shortfall risk]\label{Utility-based shortfall risk example} The utility-based shortfall risk measure with  loss function $\ell$ was introduced in  \cite{FoellmerSchied2002} as
\begin{equation}\label{utility-based shortfall risk measure}
\rho(X):=\inf\{m\in\bR\,:\,\bE[\,\ell(-X-m)\,]\le x_0\}
\end{equation}
for $X\in L^\infty$, where $\ell:\bR\to\bR_+$ is convex, nondecreasing, not identically constant, and $x_0$ belongs to the interior of $\ell(\bR)$; see also Section 4.9 in \cite{FoellmerSchied2011}. By taking $\ell(x)=e^{\beta x}$ and $x_0=1$ we recover the entropic risk measure \eqref{entropic rm eq}. In the general case, we can define a finite Young function $\Psi(x):=\ell(x)-\ell(0)$ for $x\ge0$. With this choice,  $\rho(X)$ is well-defined and finite for each $X\in H^\Psi$. Indeed, we have
\begin{equation}\label{utility-based shortfall risk measure finite}
0\le \ell(-X-m)\le \frac12\ell(-2X)+\frac12\ell(-2m)\le \frac12\Psi(2|X|)+\frac12\Psi(2|m|)+\ell(0),
\end{equation}
which implies that $\bE[\,\ell(-X-m)\,]$ is finite for $m\in\bR$ and  $X\in H^\Psi$. It is now easy to see that $\rho$ is in fact a convex risk measure on $H^\Psi$. But when $\Psi$ does not satisfy the $\Delta_2$-condition \eqref{Delta2}, then $\rho(X)$ need not be finite for each $X\in L^\Psi$. \hfill$\diamondsuit$
\end{example}

\subsection{Consistency}\label{Consistency section}

Let $\rho$ be a law-invariant convex risk measure on $H^\Psi$, where $\Psi$ is a finite Young function. We start by discussing the strong consistency of estimating the risk $\rho(X)$ from a stationary and ergodic sequence $X_1,X_2,\dots$  in  $ H^\Psi$ (see \cite[Section 6.7]{Breiman} for the definition of a stationary and ergodic process) in the sense that estimators converge a.s. This is a natural question if one wishes to estimate $\rho(X)$ from historical data or from Monte Carlo simulations, where $X$ is a random variable with the same law as $X_i$. Recall that  every i.i.d.~sequence is stationary and ergodic, and that ergodicity is implied by various mixing conditions.

We denote by
\begin{equation}\label{empirical distribution}
\wh m_n:=\frac1n\sum_{i=1}^n\delta_{X_i}
\end{equation}
the empirical distribution of $X_1,\dots, X_n$ and by
\begin{equation}\label{rho estimator}
\wh\rho_n:=\cR_\rho(\wh m_n)
\end{equation}
the corresponding estimate for $\rho(X)$. In the special case when $\rho$ is a coherent distortion risk measure, the estimator $\wh\rho_n$ has the form of an L-statistic, and so the methods and results from van Zwet \cite{vanZwet}, Gilat and Helmers \cite{GilatHelmers1997}, and Tsukahara \cite{Tsukahara}, become applicable. Our next result, however,  is valid for a general law-invariant convex risk measure $\rho$. 

\medskip

\begin{theorem}\label{Consistency thm}Suppose that $\rho$ is a law-invariant convex risk measure on $H^\Psi,$ and  let $X_1,X_2,\dots$ be a stationary and ergodic sequence of random variables with the same law as  $X\in H^\Psi$. Then \eqref{rho estimator} is a strongly consistent estimator for $\rho(X)$ in the sense that $\wh\rho_n\to\rho(X)$ $\bP$-a.s.
\end{theorem}

\medskip

For illustration, many GARCH processes are strictly stationary and ergodic (and even $\beta$-mixing); see, for instance, \cite{Boussama,Nelson}. Results on strong consistency in the case where $X_1,X_2,\ldots$ is a (not necessarily stationary) strongly mixing sequence of identically distributed random variables can be found in \cite{Zaehle2013}. Results on weak consistency and on the rate of weak convergence can be found in \cite{BelomestnyKraetschmer} for very general coherent risk measures and strong mixing, and in \cite{BeutnerZaehle2010} for a certain class of distortion risk measures and more general data dependencies. In Theorem \ref{Consistency thm} it is essential that the sequence $X_1,X_2,\dots$ satisfies a strong law of large numbers such as Birkhoff's ergodic theorem or Kolmogorov's law of large numbers. Perhaps surprisingly, it will in general not   suffice to take just any reasonable estimating sequence $(\wh\mu_n)$ for $\bP\circ X^{-1}$ to obtain the consistency $\cR_\rho(\wh\mu_n)\to\rho(X)$. This is due to the possible failure of continuity of the map $\mu\mapsto\cR_\rho(\mu)$ when $\Psi$ is not chosen suitably.  We will give a precise meaning to this in our Theorem \ref{continuity of convex risk measures}, where we analyze the continuity properties of the map $\mu\mapsto\cR_\rho(\mu)$. These continuity properties will also  be crucial for our subsequent discussion of the robustness of the estimators \eqref{rho estimator}.

\subsection{Continuity properties of $\bm\cR_{\bm\rho}$}\label{psi-weak topology and qualitative robustness}

The basic issue when discussing the continuity  of $\cR_\rho$ can already be observed in Example \ref{mean example}. There the map $\cR_\rho(\mu)=-\int x\,\mu(dx)$ is not continuous with respect to the standard weak topology of measures.  We therefore need to use a stronger topology, a fact that was already observed in  \cite{Weber}.
More precisely, we will consider the \emph{$\psi$-weak topology} associated with a
\emph{weight function} $\psi$, i.e., a continuous function $\psi:\R\to[0,\infty)$
satisfying $\psi\ge 1$ outside some compact set.
We denote by $\cM_1^\psi:=\cM_1^\psi(\bR)$ the class of all probability measures $\mu$ on $\bR$ for which $\int\psi\,d\mu<\infty$. It coincides with the set $\cM_1:=\cM_1(\bR)$ of all probability measures on $\bR$ if and only if $\psi$ is bounded.

Furthermore, $C_\psi(\R)$ will denote the space of all continuous functions $f$ on $\R$ for which $\sup_{x\in\bR}|f(x)/(1+\psi(x))|<\infty$. The {\em $\psi$-weak topology} on ${\cal M}_1^\psi$ is the coarsest topology for which all mappings $\mu\mapsto\int f\,d\mu$ with $f\in C_\psi(\R)$  are continuous; cf.\ Section A.6 in \cite{FoellmerSchied2011}. Clearly, the $\psi$-weak topology is finer than the weak topology, and the two topologies coincide if and only if $\psi$ is bounded; see Appendix \ref{psi weak Appendix} for details. When $\Psi$ is a finite Young function, then $\Psi(|\cdot|)$ is a weight function, and we will simply write $\cM_1^\Psi$ in place of $\cM_1^{\Psi(|\cdot|)}$. We will also use the term $\Psi$-weak convergence instead of $\Psi(|\cdot|)$-weak convergence etc. We recall the notation
$$
\cM(H^\Psi)=\big\{\bP\circ X^{-1}\,:\,X\in H^\Psi\big\}
$$
for the class of all laws  of random variables $X\in H^\Psi$.

\medskip

\begin{remark}For any finite Young function $\Psi$, the identity  $\cM(H^\Psi)=\cM_1^\Psi$ holds if and only if $\Psi$ satisfies the $\Delta_2$-condition \eqref{Delta2}.  Indeed, since the underlying probability space is atomless, $\cM_1^\Psi$ coincides with the set of the laws of all random variables $X$ with $\bE[\,\Psi(|X|)\,]<\infty$.  But by \cite[Theorem 2.1.17]{EdgarSucheston1992} this class of  random variables coincides with $H^\Psi$ if and only if the  $\Delta_2$-condition holds. \hfill$\diamondsuit$
\end{remark}

\medskip

\begin{theorem}\label{continuity of convex risk measures}
For a finite Young function  $\Psi$ the following conditions are equivalent.
\begin{enumerate}
\item For every law-invariant convex risk measure $\rho$ on $H^\Psi$, the map $\cR_\rho:\cM(H^\Psi)\to\bR$ is continuous  for  the $\Psi $-weak topology.
\item  $\Psi$ satisfies the $\Delta_2$-condition \eqref{Delta2}.
\end{enumerate}
\end{theorem}

\medskip

\begin{remark}Fix   $p\in [1,\infty)$ and let $\Psi_p(x)= x^{p}/p$.
According to \cite[Theorem 7.12]{Villani2003}, the $\Psi_p $-weak topology  is generated by the {\it Wasserstein metric} of order $p$,
$$
    d_{W_{p}}(\mu,\nu)\,:=\,\inf\Big\{\Big(\int |x - y|^p\,\pi(dx,dy)\Big)^{1/p}:\,\pi\in {\cal M}_1(\R\times\R)\mbox{ with marginals }\mu,\nu \Big\}.
$$
Since $\Psi_p$ satisfies the $\Delta_2$-condition (\ref{Delta2}), Theorem \ref{continuity of convex risk measures} implies that $\cR_\rho$ is continuous  with respect to $d_{W_{p}}$ whenever $\rho$ is a law-invariant convex risk measure on $L^p$. A corresponding result for $p=\infty$ is stated in \cite[Lemma 2.4]{Weber}.
{\hspace*{\fill}$\Diamond$\par\bigskip}
\end{remark}

\medskip

In Theorem \ref{continuity of convex risk measures}, a risk measure $\rho$ is given on some Orlicz heart $H^\Psi$, and it is shown that $\cR_\rho$ is continuous with respect to the $\Psi $-weak topology when $\Psi$ satisfies the $\Delta_2$-condition \eqref{Delta2}. But one could ask whether  $\cR_\rho$ is even continuous with respect to a  weaker topology. For instance, this would  be the case when $\rho$ can be extended to a law-invariant convex risk measure on a larger Orlicz heart $H^\Phi\supset H^\Psi$.

To address this question, we consider the generic situation in which $\rho$ is a law-invariant convex risk measure on $L^\infty$ and let
\begin{equation}\label{risk measure extension}
\bbar\rho:L^1\longrightarrow\bR\cup\{+\infty\}
\end{equation}
denote the unique  extension of $\rho$ that is convex, monotone,  and lower semicontinuous with respect to the $L^1$-norm. The existence of such an extension was proved in \cite{FilipovicSvindland}. When $\bbar\rho$ is finite on some Orlicz heart $H^\Psi$ with finite $\Psi$, it will be continuous on $H^\Psi$ with respect to the corresponding Luxemburg norm by \cite[Theorem 4.1]{CheriditoLi2009}, and so it will also be cash additive on $H^\Psi$ when $\rho$ is cash additive on $L^\infty$.

\medskip

\begin{theorem}
\label{extension}Suppose that $\rho$ is a  law-invariant convex risk measure on $L^\infty$.
Let furthermore $ \Psi$ be a finite Young function satisfying the $\Delta_2$-condition \eqref{Delta2}.
 Then the following
conditions are equivalent.
\begin{enumerate}
\item $\bbar\rho$ is finite on $H^ \Psi$.
\item ${\cal R}_{\overline{\rho}}$ is finite and continuous for the $\Psi$-weak topology on ${\cal M}(H^\Psi)$.
\item  ${\cal R}_{\rho}$ is finite and continuous for the $\Psi$-weak topology on ${\cal M}(L^\infty)$.
\item  If $(X_n)$ is a sequence in $L^\infty$ with $\|X_n\|_ \Psi\to0$, then $\rho(X_n)\to\rho(0)$.
\end{enumerate}
\end{theorem}

\subsection{Qualitative and comparative robustness}\label{qualitative robustness}

Informally, Hampel's classical concept of qualitative robustness of an estimator requires that  a small change in the law of the data results in only small changes in the law of the estimator. For a precise statement, it will be convenient to assume that the data arises from an i.i.d.~sequence of random variables $(X_i)$. We can then assume without loss of generality that the underlying probability space has a product structure: $\Omega=\bR^\bN$, $X_i(\omega)=\omega(i)$ for $\omega\in\Omega$ and $i\in\mathbb{N}$, and $\mathcal{F}:=\sigma(X_1,X_2,\dots)$.
For any Borel probability measure $\mu$ on $\R$, we will denote by
$$
    \mathbb{P}_\mu:=\mu^{\otimes\mathbb{N}}
$$
the corresponding product measure under which the $X_i$ are i.i.d.~with common distribution $\mu$. Then $(\Omega,\cF,\bP_\mu)$ is atomless unless $\mu$ is a Dirac measure. We will retain the definitions \eqref{empirical distribution} and \eqref{rho estimator} for the empirical distribution $\wh m_n$ and for the corresponding risk estimate $\wh\rho_n$ in this setting. We can now state the following version of Hampel's definition of qualitative robustness, suitably adapted to our more general setting. Since $\wh\rho_n$ as the plug-in estimator is determined by the risk functional $\cR_\rho$, we will refer to qualitative robustness of the sequence $(\wh\rho_n)$ simply as qualitative robustness of the risk functional $\cR_\rho$.

\medskip

\begin{definition}[Qualitative robustness]
Let $\cN\subset\cM_1$ be a set of probability measures.  Let furthermore $d_A$ be a metric on $\cN$ and $d_B$ be a metric on $\cM_1$. Then the risk functional $\cR_\rho$ is called \emph{robust on $\cN$  with respect to $d_A$ and $d_B$} if for all $\mu\in\cN$ and $\eps>0$ there exists $\delta>0$ and $n_0\in\bN$ such that
\begin{equation}\label{robustness def}
\nu\in\cN,\ d_A(\mu,\nu)\le\delta\quad \Longrightarrow\quad d_B\big(\bP_\mu\circ\wh\rho_n^{~-1},\bP_\nu\circ\wh\rho_n^{~-1})\le\eps\quad\text{for $n\ge n_0$.}
\end{equation}
\end{definition}

\medskip

In the classical literature on qualitative robustness \cite{Hampel,Huber,Cuevas1988,Mizera2010} and in \cite{Contetal}, the distances $d_A$  and $d_B$ are  chosen so that they generate the  weak topology  of measures. For instance, they are taken as the Prohorov or L\'evy metrics, and $\cN$ is taken as $\cM_1$.  It is a consequence of Hampel's celebrated theorem that such a choice essentially limits the concept of  robustness to risk functionals  that are continuous for the weak topology; see, e.g., \cite[Theorem 2.21]{Huber}.

Insisting on metrics for the weak topology when assessing the robustness of a risk functional $\cR_\rho$
is problematic for a number of
reasons.

First, two distributions $\mu$ and $\nu$ can be rather close with respect to a distance $d_A$ for the weak topology, but still have completely different tail behavior. In this case, robustness \eqref{robustness def} requires that $\cR_\rho(\mu)$ is \emph{insensitive to the tail behavior} of $\mu$. This can be seen immediately, e.g., from the definition of the L\'evy metric,
\begin{equation}\label{definition Levy metric}
d_{\text{L\'evy}}(\mu,\nu)=\inf\big\{\eps>0\,:\,F_\mu(x-\eps)-\eps\le F_\nu(x)\le F_\mu(x+\eps)+\eps\text{ for all $x$}\big\},
\end{equation}
where $F_\mu$ and $F_\nu$ are the distribution functions for $\mu$ and $\nu$. To illustrate this point, recall that  the L\'evy metric metrizes the weak topology on $\mathcal{M}_1$  and that the compactly supported probability measures are dense in $\mathcal{M}_1$ with respect to weak convergence. Hence, for every $\mu\in\cM_1$ and every $\eps>0$ there exists a compactly supported $\nu\in\cM_1$ such that $d_{\text{L\'evy}}(\mu,\nu)<\eps$.  Clearly, $\mu$ can have arbitrary tail behavior, whereas the tail behavior of $\nu$ is trivial. 
 In the recent years of financial crisis, it has become apparent, though, that a misspecification of the tail behavior of a P\&L can lead to a dramatic underestimation of the associated risk.

Second,  linking the terminology of robustness exclusively to the weak topology generates a sharp but artificial  division of risk functionals into the class of those that are called \lq\lq robust" and another class of those that are  called \lq\lq not robust".  The first class contains risk functionals such as Value at Risk which are insensitive to the tail behavior of P\&Ls, whereas the second class contains, e.g., the ordinary expectation and all law-invariant coherent risk measures  \cite{Contetal}.
This classification thus seems to provide a strong argument in favor of Value at Risk and against coherent risk measures.  We will see, however, that the distinction  between \lq\lq robust"  and \lq\lq non-robust" risk functionals is artificial because  there is actually a full continuum of possible degrees of robustness beyond the classical concept. This new look at robustness will then help us to bring the argument against coherent risk measures back into perspective: robustness is not lost entirely but only to some degree when Value at Risk is replaced by a coherent risk measure such as Average Value at Risk.

It was observed in \cite{Kraetschmer et al 2012} that the basic idea for addressing this problem is to choose suitable metrics in \eqref{robustness def}. For $d_B$ we will take the  Prohorov distance, which is defined as
$$
    d_{\mbox{\scriptsize{\rm Proh}}}(\mu,\nu)\,:=\,\inf\big\{\varepsilon>0\,:\,\mu(A)\le\nu(A^\varepsilon)+\varepsilon\mbox{ for all }A\in{\cal B}(\bR)\big\},
$$
where  $A^\varepsilon:=\{x\in \bR:\,\inf_{a\in A}|x-a|\le\varepsilon\}$ is the $\varepsilon$-hull of $A$.
 Note that $d_{\mbox{\scriptsize{\rm Proh}}}\ge d_{\text{L\'evy}}$, so our choice $d_B:= d_{\mbox{\scriptsize{\rm Proh}}}$ gives a stronger notion of robustness than the choice $d_B:= d_{\text{L\'evy}}$, which would just as well be possible.
 For $d_A$, we will choose the {\em Prohorov $\psi$-metric},
\begin{equation}\label{dpsi Definition Eq}
    d_{\psi}(\mu,\nu)\,:=\,d_{\text{Proh}}(\mu,\nu)+\Big|\int\psi\,d\mu-\int\psi\,d\nu\Big|,\qquad\mu,\nu\in{\cal M}_1^\psi,
\end{equation}
which by Lemma \ref{psi-weak-lemma} metrizes the $\psi$-weak topology on ${\cal M}_1^\psi$ for a given weight function $\psi$.  Also in \eqref{dpsi Definition Eq} we could have  replaced $d_{\mbox{\scriptsize{\rm Proh}}}$ by $d_{\text{L\'evy}}$, but here the advantage of $d_{\mbox{\scriptsize{\rm Proh}}}$ is that it extends  to a multivariate setting; see \cite{Kraetschmer et al 2012}. When $\psi(x)=\Psi(|x|)$ for some finite Young function $\Psi$, we will simply write $d_\Psi$ instead of $d_{\Psi(|\cdot|)}$.
We emphasize that, due to the summand $|\int\psi\,d\mu-\int\psi\,d\nu|$ on the right-hand side of (\ref{dpsi Definition Eq}), it strongly depends on the choice of the weight function $\psi$ in which extent the metric $d_\psi$ penalizes deviations in the tails.

We will also need the following notion.

\medskip

\begin{definition}\label{def uniformly psi integrating}
Let $\psi$ be a weight function. A set ${\cal N}\subset {\cal M}_1^\psi$ is called {\em uniformly $\psi$-integrating} when
\begin{equation}\label{def uniformly psi integrating - eq}
    \lim_{M\to\infty}\,\sup_{\nu\in{\cal N}}\,\int_{\{\psi\ge M\}}\psi\,d\nu\,=\,0.
\end{equation}
\end{definition}

\medskip

When $\psi$ is bounded,  every set $\cN\subset\cM_1^\psi=\cM_1$ is uniformly $\psi$-integrating.
But when $\psi$ is not bounded, then the uniform continuity  in \eqref{robustness def} will typically hold only on uniformly $\psi$-integrating sets $\cN$. Let us therefore introduce the following concept of robustness.

\medskip

\begin{definition}\label{def psi robustness of rf}
Let $\psi$ be a weight function and $\cM\subset\cM_1^\psi$. A risk functional  $\cR_\rho$ is called  \emph{$\psi$-robust on  $\cM$} when $\cR_\rho$ is robust with respect to  $d_{\psi}$ and $ d_{\mbox{\scriptsize{\rm Proh}}}$ on every uniformly $\psi$-integrating set $\cN\subset\cM$.
\end{definition}

\medskip

Notice that the classical notion of qualitative robustness is recovered in the case $\psi\equiv 1$. At first glance one might wonder why in Definition \ref{def psi robustness of rf} robustness of $\cR_\rho$ with respect to  $d_{\psi}$ and $d_{\mbox{\scriptsize{\rm Proh}}}$ is required  on every {\em uniformly $\psi$-integrating} set $\cN\subset\cM$. On the other hand, it seems to be ambitious to expect that robustness of $\cR_\rho$ with respect to  $d_{\psi}$ and $d_{\mbox{\scriptsize{\rm Proh}}}$ on ${\cal N}$ can hold without any condition that ensures that the tails of the probability measures $\nu\in{\cal N}$ do not differ too much. In Remark \ref{remark on UGC and Chung-LLN} below we give a motivation for restricting ourselves to uniformly $\psi$-integrating sets ${\cal N}$. Also notice that the condition of uniformly $\psi$-integrating is not very restrictive. For instance, every subset $\cN\subset\cM$ which is relatively compact for the $\psi$-weak topology is uniformly $\psi$-integrating; cf.\ Lemma \ref{charact of compact sets}. Anyway, we can now state the following preliminary result.

\medskip

\begin{proposition}\label{psi linear growth Prop}Let $\cR_\rho$ be the risk functional associated with a law-invariant convex risk measure $\rho$ on $L^\infty$. When $\psi:\bR_+\to(0,\infty)$ is a nondecreasing function such that $\cR_\rho$ is $\psi(|\cdot|)$-robust on $\cM(L^\infty)$, then $\psi$ has at least linear growth: $\liminf_{x\ua\infty}\psi(x)/x>0$.
\end{proposition}

\medskip

The significance of the preceding proposition is that it allows us to essentially limit the analysis of the $\psi$-robustness of risk functionals to weight functions  $\psi(x)=\Psi(|x|)$ arising from a Young function $\Psi$. In this context, we have the following result.

\medskip

\begin{theorem}\label{robust corollary}For a finite Young function $\Psi$, the following conditions are equivalent.
\begin{enumerate}
\item  For every law-invariant convex risk measure $\rho$  on $H^\Psi$,  $\cR_\rho$ is $\Psi $-robust on $\cM(H^\Psi)$.
\item $\Psi$ satisfies the $\Delta_2$-condition \eqref{Delta2}.
\end{enumerate}
\end{theorem}

\medskip

As in Theorem \ref{extension}, $H^\Psi$ may not be the \lq\lq canonical"\ space for $\rho$ in the sense that $\rho$ can be extended to a larger space. Such a situation has an impact on the robustness of $\rho$ as explained in the next result. By $\bbar\rho$ we denote again the extension \eqref{risk measure extension}.
\medskip

\begin{theorem}\label{Linfty robust thm}Let $\Psi$ be a finite Young function satisfying the $\Delta_2$-condition \eqref{Delta2}. For  a law-invariant convex risk measure $\rho$ on $L^\infty$, the following conditions are equivalent.
\begin{enumerate}
\item $\cR_{\bbar\rho}$ is $\Psi $-robust  on $\cM_1^\Psi$.

\item $\cR_\rho$ is $\Psi $-robust   on  $\cM(L^\infty)$.
\item  $\bbar\rho$ is finite on $H^\Psi$.
\end{enumerate}
\end{theorem}

\medskip
As a consequence of Theorem \ref{Linfty robust thm} along with Theorem \ref{extension}, robustness of a risk functional derived from a risk measure $\rho$ defined on an Orlicz heart may be reduced to continuity at $0$ of the restriction $\rho|_{L^\infty}$ of $\rho$ to $L^{\infty}$. More precisely, for any risk measure on the Orlicz heart $H^{\Psi}$ associated with a finite Young function $\Psi$ satisfying the $\Delta_2$-condition \eqref{Delta2}, we have that
$$
\cR_{\rho}~\mbox{is}~\Psi\mbox{-robust}\text{ if and only if }\rho|_{L^{\infty}}\mbox{ is continuous at}~0~\mbox{w.r.t.}~\|\cdot\|_{\Psi}.
$$
The most important aspect of Theorem \ref{Linfty robust thm} is that it allows us to study the robustness properties of a given risk functional on  $\cM(L^\infty)$ rather than on its full domain. Since any risk functional that arises from a law-invariant convex risk measure is defined on $\cM(L^\infty)$, we can thus compare two risk functionals in regard to their degree of robustness.

\medskip

\begin{definition}[Comparative robustness]\label{Comparative robustness def}Suppose that  $\rho_1$ and $\rho_2$ are two law-invariant convex risk measures on $L^\infty$. We will say that $\rho_1$ is \emph{at least as robust as $\rho_2$} if the following implication holds. When $\Psi$ is a finite Young function satisfying the $\Delta_2$-condition \eqref{Delta2}, and $\cR_{\rho_2}$ is $\Psi$-robust  on $\cM(L^\infty)$, then $\cR_{\rho_1}$ is $\Psi$-robust  on $\cM(L^\infty)$. When, in addition, there is a 
finite $\Psi$ such that $\cR_{\rho_1}$ is $\Psi$-robust  on $\cM(L^\infty)$ but $\cR_{\rho_2}$ is not, then we will say that $\rho_1$ is \emph{more robust than $\rho_2$}.
\end{definition}

\medskip

We immediately get the following corollary.

\medskip

\begin{corollary}
\label{robustcomparison}
For two  law-invariant convex risk measures $\rho_1$ and $\rho_2$ on $L^\infty$, the following conditions are equivalent.
\begin{enumerate}
\item $\rho_1$ is at least as robust as $\rho_2$.
\item When the finite Young function $\Psi$ satisfies the $\Delta_2$-condition \eqref{Delta2} and $\bbar\rho_2$ is finite on $H^\Psi$, then $\bbar\rho_1$ is also finite on $H^\Psi$.
\end{enumerate}
\end{corollary}

\medskip
\begin{example}
Since $H^{\Psi}\subset L^{1}$ for any finite Young function $\Psi,$ the ordinary expectation $\rho_1(X):=\ex[\,-X\,]$ is as least robust as any other law-invariant convex risk measure $\rho_{2}$ on $L^{\infty}.$ Furthermore, obviously $\rho_{1}$ is more robust than every risk measure $\rho_{2}$ defined in (\ref{Risk measure based on one-sided moments - eq}) with $p > 1$ and $a>0.$
\hfill$\diamondsuit$
\end{example}

\medskip

Instead of comparing the robustness of two risk measures with respect to all possible weight functions $\psi$, it makes sense to restrict the attention to the weight functions $\psi_p(x)=|x|^p/p$ for $0< p<\infty$. This leads to the following notion that was first proposed in \cite{Kraetschmer et al 2012} within a more general  context.

\medskip

\begin{definition}[Index of qualitative robustness]\label{Index of qualitative robustness}
Let $\rho$ be a law-invariant convex risk measure on $L^\infty$. The associated \emph{index of qualitative robustness} is defined as
$$
\text{iqr}(\rho) = \Big(\inf\big\{p\in (0,\infty)\,:\,
   \text{$\cR_\rho$
    is $\psi_p$-robust on $\cM(L^\infty)$}\big\}\Big)^{-1}.
$$\end{definition}

\medskip

It follows from Proposition \ref{psi linear growth Prop} that any law-invariant convex risk measure $\rho$ satisfies $\text{iqr}(\rho) \le1$. Thus, Theorem \ref{Linfty robust thm} implies that
\begin{equation}\label{representation of index of qualitative robustness}
\text{iqr}(\rho) = \Big(\inf\big\{p\in [1,\infty)\,:\,
   \text{$\bbar\rho$
    is finite on $L^p$}\big\}\Big)^{-1}.
\end{equation}



\medskip

\begin{example}[Risk measure based on one-sided moments]\label{Risk measure based on one-sided moments - part two}
For the risk measure $\rho$ defined in (\ref{Risk measure based on one-sided moments - eq}) with $p\in[1,\infty)$ and $a>0$ we obviously have ${\rm iqr}(\rho)=1/p$. \hfill$\diamondsuit$
\end{example}

\subsection{Index of qualitative robustness for  distortion risk measures}\label{distortion index}

We now turn to the  important example class of \emph{distortion risk measures}
 defined as
\begin{equation}\label{definition of distortion rm}
    \rho_g(X):=\int_{-\infty}^0 g(F_X(y))\,dy-\int_0^\infty \big(1-g(F_X(y))\big)\,dy,
\end{equation}
where $F_X$ denotes the distribution function of $X$, and $g$ is a nondecreasing function such that $g(0)=0$ and $g(1)=1$; see e.g. \cite{Yaari,Denneberg1990,Wang1996,Kusuoka,FoellmerSchied2011}.  It is a consequence of a theorem by Dellacherie and Schmeidler that $\rho_g$  is a law-invariant convex risk measure on $L^\infty$ if and only if $g$ is concave (see Proposition 4.75 and Theorem 4.94 in \cite{FoellmerSchied2011} for a proof).  In this case, $\rho_g$ is even coherent and can be represented as
\begin{equation}\label{spectral representation}
\rho_g(X)=g(0+)\esssup(-X)+\int_0^1\vatr_t(X) g'_+(t)\,dt,\qquad X\in L^\infty,
\end{equation}
where $g'_+$ is the right-hand derivative of $g$ and  $\vatr_t(X):=-\inf\{y:F_X(y)>t\}$ denotes the Value at Risk at level $t$; see, e.g., \cite[Theorem 4.70]{FoellmerSchied2011}.
It was pointed out in \cite{Contetal} that in this case $\rho_{g}$ cannot be qualitatively $\eins$-robust. On the other hand, the risk functionals of two different concave distortion risk measures may have a rather different behavior in regard to qualitative robustness.  This is the content of the following result, which in its statement  uses the Orlicz space $L^\Phi[0,1]$ over the standard Lebesgue space on the unit interval with respect to a 
Young function $\Phi$.

\medskip

\begin{proposition}\label{distortion iqr Prop} Suppose that $g$ is concave and continuous and let furthermore $\Psi$ be a finite Young function. Then the following conditions are equivalent.
\begin{enumerate}
\item $\bbar\rho_g$ is finite on $H^\Psi$.
\item  $\bbar\rho_g$ is finite on $L^\Psi$. 
\item $g'_+\in L^{\Psi^*}[0,1]$,
where  $\Psi^*(y)=\sup_{x\ge0}(xy-\Psi(x))$ is the conjugate Young function of $\Psi$.
\end{enumerate}
In particular, when $\Psi$ satisfies the $\Delta_2$-condition \eqref{Delta2}, then $\cR_{\rho_g}$ is $\Psi $-robust   on  $\cM(L^\infty)$ if and only if condition {\rm (c)} is satisfied. Moreover, 
$$\text{\rm iqr}(\rho_g)=\frac{q^*-1}{q^*}\qquad\text{where}\qquad
q^*=\sup\Big\{q\ge1\,\Big|\,\int_0^1(g'_+(t))^q\,dt<\infty\Big\}.
$$
\end{proposition}

\medskip

\begin{example}[Average Value at Risk]The risk measure Average Value at Risk at level $\alpha\in(0,1)$, $  \avatr_\alpha$, is given in terms of the concave distortion function $g_1(t)=(t/\alpha)\wedge1$ (see \cite[Example 4.71]{FoellmerSchied2011}). $ \avatr$ is also called Expected Shortfall, Conditional Value at Risk, or TailVaR. Since $g'_1$ is bounded, it  follows from Proposition \ref{distortion iqr Prop} that ${\rm iqr}(\avatr_\alpha)=1$.

More generally, we can consider the distortion function $g_\beta(t)=(t/\alpha)^\beta\wedge1$ for some $\beta\in(0,1]$.  It is easy to see that the corresponding risk measure $\rho_{g_\beta}$ has ${\rm iqr}(\rho_{g_\beta})=\beta$. It follows already from this simple example that  distortion risk measures   cover the whole possible range of our index of qualitative robustness.
\hfill$\diamondsuit$
\end{example}

\medskip

\begin{example}[MINMAXVAR] In \cite{ChernyMadan,ChernyMadan2} the concave distortion risk measures MINVAR, MAXVAR, MINMAXVAR, and MAXMINVAR were introduced. They play an important role in conic finance \cite{ChernyMadan2}. For instance, MINMAXVAR is defined in terms of the concave distortion function
$$g_{\lambda,\gamma}(t)=1-(1-t^{\frac1{1+\lambda}})^{1+\gamma},
$$
where $\lambda$ and $\gamma$ are nonnegative parameters. An easy computation shows that $g'_{\lambda,\gamma}(t)\sim c\cdot t^{-\frac\lambda{1+\lambda}}$ as $t\downarrow0$, and so we have $\text{iqr}(\text{MINMAXVAR})=\frac1{1+\lambda}$.  \hfill$\diamondsuit$

\end{example}

\section{Some general results}

In this section, we will state and prove some theorems that are crucial ingredients for the proofs of the results from Section \ref{main results section}. These theorems and propositions are interesting in their own right and stated in the general contexts of robust statistics and $\psi$-weak convergence.

\subsection{Hampel-type theorems for the $\bm\psi$-weak topology}\label{Robust section}

Hampel \cite{Hampel} introduced the notion of qualitative robustness with the weak topology in mind.  In Section \ref{qualitative robustness}  we have argued that it is necessary to replace the weak topology with a finer $\psi$-weak topology so as to obtain a more balanced picture of the robustness of a risk functional.
Such an approach was first suggested in  \cite{Kraetschmer et al 2012}. In this section, we will give versions of Hampel's theorem and its converse for the $\psi$-weak topology that are slightly stronger than the corresponding results in \cite{Kraetschmer et al 2012}. We   need them as basis for the results in Section \ref{qualitative robustness}, but we will state them here in the  framework of robust statistics rather than in the narrower context of risk functionals. For the sake of consistency with the preceding sections, we have chosen a one-dimensional setting, but we could just as well have stated our results in the even more general multivariate framework of \cite{Kraetschmer et al 2012}.

As in Section \ref{qualitative robustness}, we consider the canonical product space $\Omega:=\bR^\bN$ with Borel field $\cF$ and coordinate mappings $(X_n)$, which become i.i.d.~random variables under a product measure $\bP_\mu=\mu^{\otimes\bN}$. A \emph{statistical functional} will be a map $T:\cM\to\bR$, where $\cM\subset\cM_1$ must contain all measures of the form $\frac1n\sum_{k=1}^n\delta_{x_k}$ for $n\in\bN$ and $x_1,\dots, x_n\in\bR$. It gives rise to a sequence of estimators given by $\wh T_n:=T(\wh m_n)$, where $\wh m_n$ is the empirical distribution of $X_1,\dots, X_n$ as in \eqref{empirical distribution}. Clearly, the risk functional $\cR_\rho$ associated with a  law-invariant risk measure $\rho$ is an example of a statistical functional.  Also, recall from Definition \ref{def uniformly psi integrating} the notion of a uniformly $\psi$-integrating set. The following definition is a modified version of \cite[Definition 2.1]{Kraetschmer et al 2012}.

\medskip

\begin{definition}[$\psi$-robustness]\label{def quali rob}
Let $T$ be a statistical functional and ${\cal M}$ be a subset of $ {\cal M}_1^\psi$.
 Then $T$ is called $\psi$-robust at $\mu$ in ${\cal M}$ if for each $\varepsilon>0$ and every uniformly $\psi$-integrating set ${\cal N}\subset\cM$ with $\mu\in\cN$ there are $\delta>0$ and $n_0\in\N$ such that
$$
    \nu\in{\cal N},\  d_\psi(\mu,\nu)\le\delta\quad\Longrightarrow\quad d_{\text{Proh}}(\pr_\mu\circ \widehat T_n^{-1}\,,\,\pr_\nu\circ \widehat T_n^{-1})\le\varepsilon\quad\text{for }n\ge n_0.
$$
\end{definition}
\medskip


The following theorem provides a version of Hampel's theorem that is stronger than \cite[Corollary 3.6]{Kraetschmer et al 2012}, the corresponding result in \cite{Kraetschmer et al 2012}.

\medskip

\begin{theorem}[Hampel's theorem for the $\psi$-weak topology]\label{hampel theorem}
Let $T:\cM\to\bR$ be a statistical functional where  ${\cal M}\subset {\cal M}_1^\psi$. When  $T:\cM\to\bR$ is $\psi$-weakly continuous at $\mu\in\cM$, then $T$ is $\psi$-robust at $\mu$ in $\cM$.
\end{theorem}

\begin{proof}The result will follow from \cite[Theorem 2.4]{Kraetschmer et al 2012} when we can show that every uniformly $\psi$-integrating set $\cN$ has the following \emph{uniform Glivenko--Cantelli (UGC) property:} for each $\eps>0$ and any $\delta>0,$ there is some $n_0\in\bN$ such that
\begin{equation}\label{psi UGC}
 \sup_{\nu\in{\cal N}}\,\pr_{\nu}\big[\,d_\psi(\nu,\wh m_n)\ge\delta\,\big]\,\le\,\varepsilon\qquad\text{for $n\ge n_0$.}
\end{equation}
According to \cite[Lemma 4]{Mizera2010}, the set  ${\cal M}_1$ has the UGC property for the Prohorov metric. Therefore, the UGC property \eqref{psi UGC} follows from (\ref{dpsi Definition Eq}) and the weak version of Chung's uniform (strong) law of large numbers (\cite{Chung1951}; see also \cite[Proposition A.5.1]{van der Vaart Wellner 1996}) applied to the sequence $(\psi(X_n))$ of random variables; notice that $\int\psi\,d\wh m_n=\frac{1}{n}\sum_{i=1}^n\psi(X_i)$ and $\int\psi\,d\nu=\ex_\nu[\psi(X_1)]$.\qed
\end{proof}

\medskip

\begin{remark}\label{remark on UGC and Chung-LLN}
We note that any subset ${\cal N}\subset{\cal M}_1^\psi$ possesses the UGC property with respect to $d_\psi$ in the sense of (\ref{psi UGC}) if and only if a uniform weak law of large numbers holds for the sequence $\psi(X_1),\psi(X_2),\ldots$ within ${\cal N}$ in the sense that for all $\eps>0$ and $\delta>0$ {there is} $n_0\in\bN$ such that
\begin{equation}\label{Chungs uniform WLLN}
 \sup_{\nu\in{\cal N}}\,\pr_{\nu}\Big[\,\Big|\frac{1}{n}\sum_{i=1}^n\psi(X_i)-\int\psi\,d\nu\Big|\ge\delta\,\Big]\,\le\,\varepsilon\qquad\text{for $n\ge n_0$}.
\end{equation}
This equivalence follows from the fact that ${\cal M}_1$ possesses the UGC property with respect to the Prohorov metric $d_{\mbox{\scriptsize{\rm Proh}}}$; see, for instance, \cite[Lemma 4]{Mizera2010}. We also note that the UGC property is the key for the proof of the Hampel-type criterion of Theorem \ref{hampel theorem}.
Moreover, it was shown in \cite[pp. 345f.]{Chung1951} that, at least under the additional assumption that the medians of $\psi$ under $\nu\in{\cal N}$ are uniformly bounded, the uniform law of large numbers \eqref{Chungs uniform WLLN} is actually equivalent to the fact that $\cal N$ is uniformly $\psi$-integrating. This shows that in Definition \ref{def quali rob} we may not avoid to restrict the choice of  $\nu$  to a uniformly $\psi$-integrating set $\cal N$.
\end{remark}

\medskip

The following result may be viewed as a converse of Hampel's theorem for the $\psi$-weak topology. Together with Theorem \ref{hampel theorem} and with the choice $\psi\equiv1$ and $\cM=\cM_1$ it yields the classical Hampel theorem in the form of \cite[Theorem 2.21]{Huber}.
Its statement uses the following notion of consistency: a statistical functional $T$ is called \emph{weakly consistent at $\mu\in\cM$} when $\wh T_n\to T(\mu)$ in $\bP_\mu$-probability.

\medskip


\begin{theorem}[Converse of Hampel's theorem for  the $\psi$-weak topology]\label{hampel-huber conversed}
Suppose that $T:\cM\to\bR$ is a statistical functional where $\cM\subset\cM_1^\psi$. Let $\mu\in{\cal M}$ and $\delta_0>0$ be given, and suppose that $T$ is weakly consistent at each $\nu$ in $\cM$ with $d_\psi(\nu,\mu)\le\delta_0$. When $T$ is $\psi$-robust at $\mu$ in $\cM$, then  $T:\cM\to\bR$ is $\psi$-weakly continuous at $\mu$.
\end{theorem}

\begin{proof}
We must show that $T(\mu_k)\to T(\mu)$ when $(\mu_k)$ is a sequence in $\cM$ that converges $\psi$-weakly to $\mu$.
Given such a sequence $(\mu_k)$, the set ${\cal N}:=\{\mu,\mu_1,\mu_2,\ldots\}$ is clearly compact for the $\psi$-weak topology. By Lemma \ref{charact of compact sets} we conclude that ${\cal N}$ is uniformly $\psi$-integrating. So, given $\varepsilon>0$, the $\psi$-robustness of $T$ at $\mu$ in $\cM$ implies that there are some  $\delta>0$ and $n_0\in\bN$ such that $d_{\text{Proh}}(\bP_{\mu_k}\circ \wh T_n^{-1},\bP_\mu\circ \wh T_n^{-1})\le\eps$ for all $n\ge n_0$ and $k\ge k_0$, where $k_0\in\N$ is chosen such that $d_\psi(\mu_k,\mu)\le\delta$ for all $k\ge k_0$ (recall that $d_\psi$ generates the $\psi$-weak topology). So, assuming without loss of generality $\delta\le\delta_0$, the weak consistency of $(\hat T_n)$ at $\mu_k$ and $\mu$ implies that
\begin{eqnarray*}
    |T(\mu_k)-T(\mu)|
    & = & d_{\mbox{\scriptsize{\rm Proh}}}(\delta_{T(\mu_k)},\delta_{T(\mu)})\\
    & \le & \limsup_{n\to\infty}\Big(d_{\mbox{\scriptsize{\rm Proh}}}(\delta_{T(\mu_k)},\pr_{\mu_k}\circ \hat T_n^{-1})\,+\,\varepsilon\,+\,
    d_{\mbox{\scriptsize{\rm Proh}}}(\pr_{\mu}\circ \hat T_n^{-1},\delta_{T(\mu)})\Big)\\
    & = & \varepsilon
\end{eqnarray*}
for all $k\ge k_0$. This completes the proof.\qed
\end{proof}


\subsection{Skorohod representation  for  $\bm\psi$-weak convergence}\label{Skorohod Section}

The classical Skorohod--Dudley--Wichura representation theorem states that weak convergence $\mu_n\to\mu_0$ is equivalent to the existence of random variables $X_n$ with law $\mu_n$ such that $X_n\to X_0$ almost surely. A question one may ask is whether $\psi$-weak convergence  $\mu_n\to\mu_0$ can be expressed in terms of a stronger  concept for the convergence $X_n\to X_0$. Here we are going to address this question in the context of the Orlicz spaces.

\begin{theorem}\label{SDW}
For any finite Young function $\Psi$ the following two conditions are equivalent.
\begin{enumerate}
\item
    A sequence $(\mu_n)$ in $\cM(H^\Psi)$ converges $\Psi $-weakly to some $\mu_0$ if and only if there exists a sequence  $(X_n)_{n\in\N_0}$ in $H^\Psi$ such that $X_n$ has law $\mu_n$ for each $n\in\N_0$ and $\|X_n -X_0\|_\Psi\to 0$.
\item $\Psi$ satisfies the $\Delta_2$-condition \eqref{Delta2}.
\end{enumerate}
\end{theorem}

\medskip

For proving Theorem \ref{SDW} we need the following lemma.

\medskip

\begin{lemma}\label{uniformintegrable}
Let $\Psi$ be a finite Young function satisfying the $\Delta_2$-condition \eqref{Delta2} and let $(X_{n})_{n\in\N_{0}}$ be a sequence in $H^\Psi=L^\Psi$. If the sequence $(\Psi(|X_{n}|))_{n\in\N_{0}}$ is uniformly  integrable, then the sequence $(\Psi(2^{m}|X_{n}-X_{0}|))_{n\in\N}$ is also uniformly  integrable for every $m\in\N_{0}$.
\end{lemma}
\begin{proof}
By (\ref{Delta2}), we have $C := \sup_{x\geq x_{0}}\Psi(2 x)/\Psi(x) < \infty$ for some $x_{0} > 0.$ We  proceed by induction on $m\in\N_{0}$.

First, let $m=0.$ Since $\Psi$ is nondecreasing and convex with $\Psi(0) = 0$, we obtain by the triangle inequality
\begin{eqnarray*}
    \Psi(|X_{n}- X_{0}|)
    & \le &
    \frac{1}{2}\Big(\Psi(2 |X_{n}|) + \Psi(2 |X_{0}|)\Big) \\
    & = &
    \frac{1}{2}\Big(\Psi\big(2\,\eins_{[0,x_{0}]}(|X_{n}|) |X_{n}|\big) +
    \Psi\big(2\,\eins_{[0,x_{0}]}(|X_{0}|)|X_{0}|\big)\Big)\\
    & & +\,
    \frac{1}{2}\Big(\Psi\big(2\,\eins_{(x_{0},\infty)}(|X_{n}|) |X_{n}|\big) +
    \Psi\big(2\,\eins_{(x_{0},\infty)}(|X_{0}|)|X_{0}|\big)\Big)\\
    & \le &
    \Psi(2 x_{0})+ \,\frac{C}{2}\Big(\Psi(|X_{n}|) + \Psi(|X_{0}|)\Big).
\end{eqnarray*}
Since the sequence $(\Psi(|X_{n}|))_{n\in\N_{0}}$ is uniformly  integrable by assumption, we may thus conclude that the sequence $(\Psi(|X_{n}-X_{0}|))_{n\in\N}$ is uniformly integrable.

Let us now suppose that $(\Psi(2^{m}|X_{n}-X_{0}|))_{n\in\N}$ is uniformly integrable for any given $m\in\N_{0}$. Following an analogous line of reasoning as in the case of $m = 0$, we may find
$$
    \Psi(2^{m+1}|X_{n}-X_{0}|) \,=\, \Psi(2\cdot 2^{m}|X_{n}-X_{0}|) \,\leq\, \Psi(2 x_{0}) + C\,\Psi(2^{m}|X_{n}-X_{0}|).
$$
Hence, $(\Psi(2^{m+1}|X_{n}-X_{0}|))_{n\in\N}$ is uniformly integrable, which completes the proof.
\end{proof}\goodbreak

\noindent{\it Proof of {\rm(b)$\Rightarrow$(a)} in Theorem \ref{SDW}.}
Let us suppose that the $\Delta_2$-condition \eqref{Delta2} holds.

We first prove that $\|X_n -X_0\|_\Psi\to 0$ implies $\mu_n\to\mu_0$ $\Psi $-weakly. By \eqref{Proposition 2.1.10 in EdgarSucheston1992},  $\|X_n -X_0\|_\Psi\to 0$ yields $\bE[\,\Psi(2|X_n-X_0|)\,]\to 0$ and $X_n\to X_0$ in probability. Convexity and monotonicity of $\Psi$ imply that
$$0\le\Psi(|X_n|)\le\frac12\Psi\big(2\big||X_n|-|X_0|\big|\big)+\frac12\Psi(2|X_0|)\le \frac12 \Psi(2|X_n-X_0|)+\frac12\Psi(2|X_0|).
$$
Hence, $\Psi(|X_n|)$ is uniformly integrable, and we obtain that
$$\int\Psi(|x|)\,\mu_n(dx)=\bE[\,\Psi(|X_n|)\,]\longrightarrow\bE[\,\Psi(|X_0|)\,]=\int\Psi(|x|)\,\mu_0(dx).$$
 Moreover, since $X_n\to X_0$ in probability the corresponding laws $(\mu_n)$ converge weakly. Now the $\Psi $-weak convergence $\mu_n\to\mu_0$ follows from Lemma \ref{psi-weak-lemma}\,(iv)$\Rightarrow$(i).

Now we prove that the $\Psi $-weak convergence $\mu_n\to\mu_0$ implies the existence of a sequence $(X_n)$ in $H^\Psi$ such that $\|X_n -X_0\|_\Psi\to 0$. Clearly, $\mu_n\to\mu_0$ weakly. By Skorohod representation there hence exists a sequence of random variables $(X_n)$ such that $X_n\to X_0$ $\bP$-a.s. The  continuity of $\Psi$ and the fact that $\Psi(0)=0$ yield that
\begin{align}
    \Psi(|X_n|) \longrightarrow \Psi(|X_0|)&\qquad \bP\mbox{-a.s.}
\label{SDW - proof - eq - 1} \\
    \Psi(k|X_n-X_0|) \longrightarrow 0&\qquad\bP\mbox{-a.s. for all $k\ge0$.}  \label{SDW - proof - eq - 2}
\end{align}
Moreover, the  $\Psi $-weak convergence   $\mu_n\to\mu_0$ implies that
\begin{equation}\label{SDW - proof - eq - 3}
    \bE[\,\Psi(|X_n|)\,]=\int\Psi(|x|)\,\mu_n(dx)\longrightarrow \int\Psi(|x|)\,\mu_0(dx)=\bE[\,\Psi(|X_0|)\,].
\end{equation}
Now,  \eqref{SDW - proof - eq - 1}, (\ref{SDW - proof - eq - 3}), and Vitali's theorem in the form of \cite[Proposition 3.12\,(ii)$\Rightarrow$(iii)]{Kallenberg1997} imply that the sequence $(\Psi(|X_n|))_{n\in\N_{0}}$ is uniformly  integrable. Applying Lemma  \ref{uniformintegrable} yields the uniform integrability of the sequence $(\Psi(k|X_n-X_0|))_{n\in\bN}$ for every $k>0$. Therefore, \eqref{SDW - proof - eq - 2} and another application of Vitali's theorem, this time in the form of \cite[Proposition 3.12\,(iii)$\Rightarrow$(ii)]{Kallenberg1997}, yield $\bE[\,\Psi(k|X_n-X_0|)\,]\to 0$ for every $k>0$, which implies $\|X_n-X_0\|_\Psi\to 0$  according to \eqref{Proposition 2.1.10 in EdgarSucheston1992}. Finally, the sequence $(X_n)$ belongs to $H^\Psi$, because under the $\Delta_2$-condition $H^\Psi$ coincides with the class of random variables $Y$ with $\bE[\,\Psi(|Y|)\,]<\infty$.
\qed

\medskip

\noindent{\it Proof of {\rm(a)$\Rightarrow$(b)} in Theorem \ref{SDW}.} Let us suppose that condition (a) in Theorem \ref{SDW} holds, but that $\Psi$ does not satisfy the $\Delta_2$-condition \eqref{Delta2}. We will show that this leads to a contradiction. Since $\Psi$ does not satisfy the $\Delta_2$-condition \eqref{Delta2} and our probability space is atomless, we have $H^\Psi\neq L^\Psi$ by \cite[Theorem 2.1.17]{EdgarSucheston1992}. Hence there exists a random variable $Y\ge 0$ such that $\bE[\,\Psi(Y)\,]<\infty$ and $\bE[\,\Psi(2Y)\,]=\infty$. We then choose $a_n>0$ such that
$$
2\bE\Big[\,\Psi\big(2\big(Y\wedge a_n\big)\big)\,\Big]\ge n+\Psi(4n)\qquad\text{for each $n$,}
$$
and let
$$X_n:=(Y-n)^+\wedge a_n.
$$
Then $X_n\in L^\infty$, and hence $X_n\in H^\Psi$ since $\Psi$ is finite. Moreover, $X_n\to X_0:=0~\mathbb{P}-$a.s., and
\begin{equation}\label{mun laws to Dirac0 Psi weakly}
0\le\bE[\,\Psi(|X_n|)\,]=\bE\big[\,\Psi((Y-n)^+\wedge a_n)\,\big]\le\bE\big[\,\Psi((Y-n)^+)\,\big] \longrightarrow 0
\end{equation}
by dominated convergence. It therefore follows from Lemma \ref{psi-weak-lemma} that $\mu_n:=\bP\circ X_n^{-1}\to\delta_0$ in the $\Psi $-weak topology.

We will show next that we cannot have $\|X_n -X_0\|_\Psi=\|X_n\|_\Psi\to 0$. Since any sequence $(\wt X_n)$ for which $\bP\circ\wt X_n^{-1}=\mu_n$ must satisfy $\|\wt X_n\|_\Psi=\|X_n\|_\Psi$, condition (a) in Theorem \ref{SDW} will thus be violated. So let us suppose by way of contradiction that $\|X_n\|_\Psi\to 0$. By  \eqref{Proposition 2.1.10 in EdgarSucheston1992}, this is equivalent to  $\bE[\,\Psi(k|X_n|)\,]\to0$ for every $k>0$. By taking $k=4$ and using the fact that the convex function $\ell(x):=\Psi(4 x^+)$ satisfies $\ell(x-y)\ge2\ell(x/2)-\ell(y)$ we obtain
\begin{eqnarray*}\bE[\,\Psi(4|X_n|)\,]&=&\bE\big[\,\Psi\big(4\big((Y-n)^+\wedge a_n\big)\big)\,\big]\ge\bE\Big[\,\Psi\Big(4\big(Y\wedge(a_n+n)-n\big)^+\Big)\,\Big]\\
&\ge&2\bE\Big[\,\Psi\Big(2\big(Y\wedge(a_n+n)\big)\Big)\,\Big]-\Psi(4n)\\
&\ge& 2\bE\Big[\,\Psi\big(2\big(Y\wedge a_n\big)\big)\,\Big]-\Psi(4n)\ge n,
\end{eqnarray*}
by construction. This is the desired contradiction. \qed

\section{Proofs of the results from Section \ref{main results section}}\label{Proof section}

\noindent{\it Proof of Theorem \ref{Consistency thm}.} Since $\bE[\,\Psi(k|X|)\,]<\infty$ for each $k>0$, Birkhoff's ergodic theorem (e.g., in the form of \cite[Theorem 6.28]{Breiman}) implies that for each $k>0$ and $\bP$-a.e. $\omega\in\Omega$
\begin{equation}\label{lln for Psi(k)}
\int \Psi(k|x|)\,\wh m_n(\omega)(dx)=\frac1n\sum_{i=1}^n\Psi(k|X_i(\omega)|)\longrightarrow\bE[\,\Psi(k|X|)\,]=\int\Psi(k|x|)\,\mu(dx),
\end{equation}
where  $\mu:=\bP\circ X^{-1}$. Moreover,  for $\bP$-a.e. $\omega\in\Omega$
\begin{equation}\label{lln for emp dist}
\wh m_n(\omega)\longrightarrow\mu\qquad\text{weakly,}
\end{equation}
due to Birkhoff's ergodic theorem and an application of \cite[Theorem 6.6]{Parthasarathy}.
Hence there exists a measurable set $\Omega_0\in\cF$ such that $\bP[\,\Omega_0\,]=1$ and such that for each $\omega\in\Omega_0$  \eqref{lln for emp dist} is satisfied and \eqref{lln for Psi(k)} holds for each $k\in \bN$. Let us fix ${\omega_0}\in\Omega_0$.  Since our probability space is atomless, standard Skorohod representation yields the existence of random variables $X^{\omega_0}$,  $(X_n^{\omega_0})$  such that $X^{\omega_0}$ has law $\mu$, $X_n^{\omega_0}$ has law $\wh m_n({\omega_0})$, and $X_n^{\omega_0}\to X^{\omega_0}$ $\bP$-a.s.  By \eqref{lln for Psi(k)},
for each $k\in\bN$,
\begin{equation}\label{convergence for k}
\bE[\,\Psi(k|X^{\omega_0}_n|)\,]=\int \Psi(k|x|)\,\wh m_n({\omega_0})(dx)\longrightarrow \int\Psi(k|x|)\,\mu(dx)=\bE[\,\Psi(k|X^{\omega_0}|)\,].
\end{equation}
Therefore the sequence $(\Psi(k|X^{\omega_0}_n|))$ is uniformly integrable for each $k\in\bN$.

Now take $a>0$ and pick $k\in\bN$ such that $k\ge 2a$. Since $\Psi$ is convex and nondecreasing, we have
$$0\le  \Psi(a|X^{\omega_0}_{n}- X^{\omega_0}|)
     \le
    \frac{1}{2}\Big(\Psi(2a |X_{n}^{\omega_0}|) + \Psi(2a |X^{\omega_0}|)\Big)\le  \frac{1}{2}\Big(\Psi(k |X_{n}^{\omega_0}|) + \Psi(k |X^{\omega_0}|)\Big).
$$
It follows that the sequence $(\Psi(a|X^{\omega_0}_{n}- X^{\omega_0}|))$ is uniformly integrable. Since clearly $\Psi(a|X^{\omega_0}_{n}- X^{\omega_0}|)\to0$ $\bP$-a.s., we get that $\bE[\,\Psi(a|X^{\omega_0}_{n}- X^{\omega_0}|)\,]\to0$ for each $a>0$ and in turn that $\|X^{\omega_0}_{n}- X^{\omega_0}\|_\Psi\to0$ due to \eqref{Proposition 2.1.10 in EdgarSucheston1992}. By \cite[Theorem 4.1]{CheriditoLi2009}   $\rho$ is continuous with respect to the Luxemburg norm $\|\cdot\|_{\Psi}$, and so
$$\wh\rho_n({\omega_0})=\cR_\rho(\wh m_n({\omega_0}))=\rho(X_n^{\omega_0})\longrightarrow\rho(X^{\omega_0})=\rho(X)
$$
for each ${\omega_0}\in\Omega_0$. \qed

\bigskip

\noindent{\it Proof of Theorem \ref{continuity of convex risk measures}.}
We first prove the implication (b)$\Rightarrow$(a) in Theorem \ref{continuity of convex risk measures}. So let us assume that $\Psi$ satisfies the $\Delta_2$-condition \eqref{Delta2} and  let $\rho$ be a convex risk measure on $H^\Psi$ with associated map $\cR_\rho$. It suffices to show sequential continuity of $\cR_\rho$ since the
$\psi$-weak topology is metrizable; cf.\ \cite[Corollary A.45]{FoellmerSchied2011}. So let us choose a sequence $(\mu_n)$ such that $\mu_n\to\mu_0$ $\Psi $-weakly.  By Theorem   \ref{SDW} there exists a sequence $(X_n)_{n\in\N_0}$ in $H^\Psi$ such that each $X_n$ has law $\mu_n$ and such that $\|X_n- X_0\|_\Psi\to0$.   But it was shown in \cite[Theorem 4.1]{CheriditoLi2009}  that $\rho$ is continuous with respect to the Luxemburg norm $\|\cdot\|_{\Psi}$ (see also \cite[Proposition 3.1]{RuSha}). Therefore,
$$\cR_\rho(\mu_n)=\rho(X_n)\longrightarrow\rho(X_0)=\cR_\rho(\mu_0),
$$
which proves the implication (b)$\Rightarrow$(a).

We now prove the implication (a)$\Rightarrow$(b) in Theorem \ref{continuity of convex risk measures}. This proof is similar to the proof of {\rm(a)$\Rightarrow$(b)} in Theorem \ref{SDW}. We assume that $\Psi$ does not satisfy the $\Delta_2$-condition \eqref{Delta2}, and we will construct a risk measure $\rho$ for which $\cR_\rho$ is not $\Psi $-weakly continuous.  This risk measure is given as the utility-based shortfall risk measure \eqref{utility-based shortfall risk measure} with convex loss function $\ell(x):=\Psi(8x^+)$. It follows as in \eqref{utility-based shortfall risk measure finite} that $\bE[\,\ell(-X-m)\,]$ is finite and well-defined for  $m\in\bR$ and $X\in H^\Psi$.

Since $\Psi$ does not satisfy the $\Delta_2$-condition \eqref{Delta2} and our probability space is atomless, we have $H^\Psi\neq L^\Psi$ by \cite[Theorem 2.1.17]{EdgarSucheston1992}. Hence there exists a random variable $Y\ge 0$ such that $\bE[\,\Psi(Y)\,]<\infty$ and $\bE[\,\Psi(2Y)\,]=\infty$. We then choose $a_n>0$ such that
\begin{equation}\label{an choice2}
4\bE\Big[\,\Psi\big(2\big(Y\wedge a_n\big)\big)\,\Big]\ge n+\ell(n/2)\qquad\text{for each $n$,}
\end{equation}
and let
$$X_n:=-\big((Y-n)^+\wedge a_n\big).
$$
Then $X_n\in L^\infty$ and hence $X_n\in H^\Psi$ since $\Psi$ is finite. As in \eqref{mun laws to Dirac0 Psi weakly} we get
$\bE[\,\Psi(|X_n|)\,]\to 0$, and so  $\mu_n:=\bP\circ X_n^{-1}\to\delta_0$ in the $\Psi $-weak topology. We now prove that the sequence $z_n:=\rho(X_n)=\cR_\rho(\mu_n)$ is unbounded, which will imply that $\cR_\rho$ is not continuous for the $\Psi $-weak topology.

To prove that the sequence $(z_n)$ is unbounded, we assume by way of contradiction that $z^*:=\sup_nz_n<\infty$. We see from \eqref{utility-based shortfall risk measure finite} and dominated convergence  that each $z_n=\rho(X_n)$ solves the equation $\bE[\,\ell(-X_n-z_n)\,]=1$. The convexity of $\ell$ implies that
$\ell(x-y)\ge2\ell(x/2)-\ell(y)$. Hence,
\begin{eqnarray*}
1&=&\bE[\,\ell(-X_n-z_n)\,]\ge \bE[\,\ell(-X_n-z^*)\,]\ge 2\bE\Big[\,\ell\Big(\frac12\big((Y-n)^+\wedge a_n\big)\Big)\,\Big]-\ell(z^*)\\
&\ge&2\bE\Big[\,\ell\Big(\frac12\big(Y\wedge (a_n+n)\big)-\frac n2\Big)\,\Big]-\ell(z^*)\\
&\ge& 4\bE\Big[\,\ell\Big(\frac14\big(Y\wedge a_n\big)\Big)\,\Big]-\ell(n/2)-\ell(z^*) \\
&=&4\bE\Big[\,\Psi\big(2\big(Y\wedge a_n\big)\big)\,\Big]-\ell(n/2)-\ell(z^*).
\end{eqnarray*}
But according to \eqref{an choice2}, the expression on the right is bounded from below by $n-\ell(2z^*)$, which yields the desired contradiction. \qed
\medskip

\noindent{\it Proof of Theorem \ref{extension}.}
The equivalence between conditions  (a) and (b)  follows from Theorem \ref{continuity of convex risk measures}. The implication (b)$\Rightarrow$(c) simply follows from the fact that $\rho$ is equal to the restriction of $\bbar\rho$ to $L^\infty$.

To prove (c)$\Rightarrow$(d), we first note that
$\mu_n:=\bP\circ X_n\in\cM(L^\infty)$ when $(X_n)$ is a sequence as in (d). Moreover,  $\|X_n\|_ \Psi\to0$ implies that $\mu_n\to\delta_0$ $ \Psi $-weakly. So it is now clear that (c) implies (d).

We now prove
(d)$\Rightarrow $(a). To this end, we will apply \cite[Theorem 4.3]{CheriditoLi2009}, which states that $\bbar\rho$ is finite on $H^ \Psi$ when $0$ belongs to the topological interior of the effective domain
of the map $\bbar\rho:H^ \Psi\to\bR\cup\{+\infty\}$. One can apply Proposition 2.18 and the subsequent remark in \cite{Farkasetal} to get the same implication when $\rho$ is not cash-additive and only cash-coercive as in Remark \ref{cash coercivity rem}.  We will therefore show that $\rho$ is finite on the centered $\eps$-ball $B_\eps:=\{X\in H^ \Psi\,:\,\|X\|_ \Psi<\eps\}$ when $\eps>0$ is small enough.

Suppose that $(X_n)$ is a sequence in $L^\infty$ such that $\|X_n\|_\Psi\to0$. Then we have $\mu_n:=\bP\circ X_n^{-1}\to\delta_0$ $\Psi$-weakly, and so $\rho(X_n)=\cR_\rho(\mu_n)\to\cR_\rho(\delta_0)=\rho(0)$.
 Thus, $\rho:L^\infty\to\bR$ is continuous with respect to $\|\cdot\|_\Psi$ at $0$. Hence, for $K>0$ given, there exists $\eps>0$ such that
$
\rho(X)\leq K$ for $X\in B_{\varepsilon}\cap L^{\infty}$.
Now let us fix $X\in B_\eps$. The negative part $X^{-}$ belongs again to $B_\eps$, and monotone convergence yields $X^-\wedge k\to X^-$ in ${L^1}$.  Using the lower semicontinuity of $\bbar\rho:L^1\to\bR\cup\{+\infty\}$ hence gives,
$$\bbar\rho(X)\le\bbar\rho(-X^-)\le\liminf_{k\ua\infty}\bbar\rho(-X^-\wedge k)=\liminf_{k\ua\infty}\rho(-X^-\wedge k)\le K.
$$
Here we have also used the monotonicity of $\bbar\rho$ in the first and the fact that $-(X^-\wedge k)$ belongs to $B_\eps\cap L^\infty$ in the final step. \qed

\bigskip

\noindent{\it Proof of Proposition \ref{psi linear growth Prop}.} We prove the assertion by way of contradiction. So let $\psi:\bR_+\to(0,\infty)$ be a nondecreasing function such that $\liminf_{x\ua\infty}\psi(x)/x=0$ and suppose  that  $\cR_\rho$ is $\psi(|\cdot|)$-robust on $\cM(L^\infty)$. By Theorem \ref{Consistency thm}, $\cR_\rho$ is strongly consistent at each $\mu\in\cM(L^\infty)$. Hence Theorem \ref{hampel-huber conversed} and the robustness of $\cR_\rho$ imply the continuity of $\cR_\rho$ on $\cM(L^\infty)$ with respect to $\psi(|\cdot|)$-weak convergence. As in the proofs of Theorems  \ref{SDW} and \ref{continuity of convex risk measures} we will construct a sequence $(\mu_n)\subset\cM(L^\infty)$ that converges $\psi(|\cdot|)$-weakly to $\delta_0$ but for which $\cR_\rho(\mu_n)\not\to\cR_\rho(\delta_0)$. To this end, we easily construct a random variable $Y\ge0$ such that $\bE[\,\psi(Y)\,]<\infty$ and $\bE[\,Y\,]=\infty$ and pick $a_n>0$ such that $\bE[\,Y\wedge a_n\,]\ge 2n$. Then $X_n:=(Y-n)^+\wedge a_n\to0$ $\bP$-a.s. and $\bE[\,\psi(X_n)\,]\to0$ by dominated convergence. Hence, $\mu_n:=\bP\circ (-X_n)^{-1}$ converge $\psi(|\cdot|)$-weakly to $\delta_0$ by Lemma \ref{psi-weak-lemma}.  However, $\bE[\,X_n\,]\ge\bE[\,Y\wedge a_n\,]-n\ge n$. Now \cite[Lemma 2.3]{Schied2004}  yields that
$$\cR_\rho(\mu_n)=\rho(-X_n)\ge\rho(\bE[\,-X_n\,])\ge\rho(- n),
$$
which shows that we cannot have  $\cR_\rho(\mu_n)\to\cR_\rho(\delta_0)$.\qed

\bigskip

\noindent{\it Proof of Theorem \ref{robust corollary}.}
(a)$\Rightarrow$(b): By Theorem \ref{Consistency thm}, $\cR_\rho$ is strongly consistent at each $\mu\in\cM(H^\Psi)$. Hence Theorem \ref{hampel-huber conversed} and the robustness of $\cR_\rho$ imply the continuity of $\cR_\rho$ on $\cM(H^\Psi)$ with respect to $\Psi$-weak convergence.  Thus,  due to Theorem \ref{continuity of convex risk measures}, $\Psi$ must satisfy the $\Delta_2$-condition \eqref{Delta2}.

(b)$\Rightarrow$(a): By Theorem \ref{continuity of convex risk measures}, $\cR_\rho$ is a continuous map on $\cM(H^\Psi)=\cM_1^\Psi$. Its $\Psi$-robustness on $\cM_1^\Psi$ therefore follows from Theorem \ref{hampel theorem}.
\qed

\bigskip

\noindent{\it Proof of Theorem \ref{Linfty robust thm}.}
The implication (a)$\Rightarrow$(b) is obvious.

(b)$\Rightarrow$(c): First, we note again that $\cR_\rho$ is strongly consistent on $\cM(L^\infty)$ by Theorem \ref{Consistency thm}. Therefore, Theorem \ref{hampel-huber conversed} and the robustness of $\cR_\rho$ imply the $\Psi $-weak continuity of $\cR_\rho$ on $\cM(L^\infty)$. Theorem \ref{extension} now yields (c).

(c)$\Rightarrow$(a):
Condition (c) implies that $\bbar\rho$ is a convex risk measure on $H^\Psi$. Hence, (a) follows by applying Theorem \ref{robust corollary}. \qed

\bigskip

\goodbreak\noindent{\it Proof of Proposition \ref{distortion iqr Prop}.} First, when $g$ is continuous we have $g(0+)=0$ and in the \lq spectral\rq\ representation  \eqref{spectral representation} the part containing  the essential supremum vanishes. Moreover, it follows from \cite{FilipovicSvindland} that  \eqref{spectral representation} remains true for $\bbar\rho_g$ and $X\in L^1$. Next, the function $f(t):=g'_+(1-t)$ is nondecreasing, and we have
$$\int_0^1\vatr_t(X) g'_+(t)\,dt=\int_0^1q_{-X}(t)f(t)\,dt,
$$
where $q_{-X}$ is a quantile function for $-X$.

Let us now show the implication (c)$\Rightarrow$(b). To this end, suppose that $f\in L^{\Psi^*}[0,1]$. 
For any $X\in L^\Psi$ we have $q_{-X}\in L^\Psi[0,1]$ because  under the Lebesgue measure on $[0,1]$, $q_{-X}$ has the same law as $-X$ under $\bP$. Thus by \cite[Proposition 2.2.7]{EdgarSucheston1992}, we get that 
\begin{eqnarray*}
\bbar\rho_g(X)=\int_0^1q_{-X}(t)f(t)\,dt\le 2 \|q_{-X}\|_{L^\Psi[0,1]}\|f\|_{ L^{\Psi^*}[0,1]}=2\|X\|_{L^\Psi}\|f\|_{ L^{\Psi^*}[0,1]}<\infty.
\end{eqnarray*}
Here $\|\cdot\|_{L^\Psi[0,1]}$ and $\|\cdot\|_{L^{\Psi^{*}}[0,1]}$ stand for the Luxemburg norms on $L^\Psi[0,1]$ and $L^{\Psi^{*}}[0,1]$ respectively.

Condition (b) trivially implies (a). So it remains to show that (a) implies (c). To this end, 
we assume that $\bbar\rho_g$ is finite on $H^\Psi$.  Since our probability space is atomless, it supports  a random variable $U$ with uniform distribution on $(0,1)$.  We will show that $Y:=f(U)$ belongs to $L^{\Psi^*}$, which in turn implies (c) since under the Lebesgue measure on $[0,1]$, $f$ has the same distribution as $Y$ under $\bP$. By \cite[Theorem 2.2.11]{EdgarSucheston1992}, $L^{\Psi^*}$ is the topological dual of the Banach space $H^\Psi$.   According to the Banach--Steinhaus theorem (or \cite[Proposition 2.2.7 with Corollary 2.2.10]{EdgarSucheston1992}) we thus have $Y\in L^{\Psi^*}$ if and only if $\bE[\, (-X)Y\,]<\infty$ for all $X\in H^\Psi$.  But for $X\in H^\Psi$, the fact that $f$ is a quantile function for $Y$ and the upper Hardy--Littlewood inequality (e.g., \cite[Theorem A.24]{FoellmerSchied2011}) imply that
$$\infty>\bbar\rho_g(X)=\int_0^1q_{-X}(t)f(t)\,dt=\int_0^1q_{-X}(t)q_Y(t)\,dt\ge \bE[\, (-X)Y\,].
$$
This concludes the proof.
\qed


\appendix
\section{Auxiliary results on the $\psi$-weak topology}\label{psi weak Appendix}

First recall from \cite[Corollary A.45]{FoellmerSchied2011} that the $\psi$-weak topology on $\cM_1^\psi(\bR)$ is separable and metrizable.
The following lemma provides some useful characterizations of the $\psi$-weak convergence; see \cite[Lemma 3.4]{Kraetschmer et al 2012} for a proof.

\begin{lemma}\label{psi-weak-lemma}
The following statements are equivalent:
\begin{itemize}
    \item[(i)] $\mu_n\to\mu$ $\psi$-weakly.
    \item[(ii)] $\int f\,d\mu_n\to\int f\,d\mu$ for every $f\in C_\psi(\R)$.
    \item[(iii)] $\int f\,d\mu_n\to\int f\,d\mu$ for every continuous $f$ with compact support and for $f=\psi$.
    \item[(iv)] $\mu_n\to\mu$ weakly and $\int \psi\,d\mu_n\to\int \psi\,d\mu$.
\end{itemize}
\end{lemma}

\medskip





The following lemma gives a transparent characterization of the $\psi$-weakly compact subsets of ${\cal M}_1^\psi$. Recall that a set ${\cal N}\subset {\cal M}_1^\psi$ is called uniformly $\psi$-integrating if it satisfies (\ref{def uniformly psi integrating - eq}).

\medskip

\begin{lemma}\label{charact of compact sets}
A set ${\cal N}\subset {\cal M}_1^\psi$ is relatively compact for the $\psi$-weak topology if and only if there exists a measurable function $\phi:\mathbb{R}\to[0,\infty)$ such that  $\phi(x)/\psi(x)\to\infty$ as $|x|\to\infty$ and such that
\begin{equation}\label{charact of compact sets - eq}
    \sup_{\nu\in{\cal N}}\int\phi\,d\nu<\infty.
\end{equation}
In this case, ${\cal N}$ is uniformly $\psi$-integrating.
\end{lemma}

\begin{proof}
The first statement is an immediate consequence of Corollary A.47 in \cite{FoellmerSchied2011}. For bounded $\psi$, the second statement is trivial. To prove the second statement for unbounded $\psi$, we assume without loss of generality that $\phi>0$. Fix $\varepsilon>0$, and denote by $K$ the left-hand side of (\ref{charact of compact sets - eq}). Choosing $M_1>0$ so large so that $\psi(x)/\phi(x)\le\varepsilon/K$ when $|x|\ge M_1$, and choosing $M_0>0$ so large so that $\psi(x)\ge M_0$ implies $|x|\ge M_1$, we obtain
\begin{eqnarray*}
    \sup_{\nu\in{\cal N}}\int\psi(x)\eins_{\{\psi(x)\ge M\}}\,\nu(dx)
    & = & \sup_{\nu\in{\cal N}}\int\phi(x)\,\frac{\psi(x)}{\phi(x)}\,\eins_{\{\psi(x)\ge M\}}\,\nu(dx)\\
    & \le & \frac{\varepsilon}{K}\,\sup_{\nu\in{\cal N}}\int\phi(x)\,\nu(dx)\\
    & = & \varepsilon
\end{eqnarray*}
for all $M\ge M_0$. That is, (\ref{def uniformly psi integrating - eq}) holds.\qed
\end{proof}

%

\noindent{\bf Acknowledgement.}
The authors thank Freddy Delbaen, Paul Embrechts, Marco Frittelli, and two anonymous referees for comments, which helped to improve a previous draft of the paper.


\parskip-0.5em\renewcommand{\baselinestretch}{0.8}



\small
\begin{thebibliography}{ccc}

\bibitem{AcerbiTasche} Acerbi, C., Tasche, D.: On the coherence of expected shortfall. Journal of Banking \& Finance {\bf 26}, 1487--1503 (2002)

\bibitem{BelomestnyKraetschmer}
Belomestny, D., Kr\"atschmer, V.: Central limit theorems for law-invariant coherent risk measures. Journal of Applied Probability {\bf 49}, 1- 21 (2012). 

\bibitem{BeutnerZaehle2010}
Beutner, E., Z\"ahle, H.: A modified functional delta method and its application to the estimation of risk functionals. Journal of Multivariate Analysis {\bf 101}, 2452--2463 (2010).

\bibitem{BiaginiFrittelli} Biagini, S., Frittelli,  M.: 
A unified framework for utility maximization problems: an Orlicz space approach.
Annals of Applied Probability {\bf 18}, 929-966 (2008).

\bibitem{BiaginiFrittelli2} Biagini, S., Frittelli, M.:  On the extension of the Namioka-Klee theorem and on the Fatou property for risk measures. 
In: Delbaen, F. et al. (eds.): Optimality and risk: modern trends in mathematical finance. The Kabanov Festschrift, pp.1--29. Springer, Berlin Heidelberg New York (2009).

\bibitem{Boussama} Boussama, F.: Ergodicity, mixing and estimation in GARCH models. Ph.D.\ Thesis, University of Paris 7 (1998).

\bibitem{Breiman} Breiman, L.: Probability (Classics in Applied Mathematics 7) SIAM, Philadelphia, PA (1991) (Corrected reprint of the 1968 original).

\bibitem{CheriditoLi2009} Cheridito, P., Li, T.: Risk measures on Orlicz hearts. Mathematical Finance {\bf 19}, 189--214 (2009)

\bibitem{ChernyMadan} Cherny, A.,   Madan, D.: New measures for performance evaluation. Review of Financial Studies {\bf 22}, 2571--2606 (2009)

\bibitem{ChernyMadan2} Cherny, A.,   Madan, D.: Markets as a counterparty: an introduction to conic finance. International Journal of Theoretical and Applied Finance, {\bf 13},   1149--1177 (2010)

\bibitem{Contetal} Cont, R., Deguest, R., Scandolo, G.: Robustness and sensitivity analysis of risk measurement procedures. Quantitative Finance, {\bf 10}, 593--606 (2010). 

\bibitem{Chung1951} Chung, K.L. (1951). The strong law of large numbers. In:  Proceedings of the Second Berkeley Symposium on Mathematical Statistics and Probability, University of California Press, Berkeley Los Angeles, pp.341--352 (1951).

\bibitem{Cuevas1988} Cuevas, A.: Qualitative robustness in abstract inference. Journal of Statistical Planning and Inference {\bf 18}, 277--289 (1988)

\bibitem{Delbaen2002} Delbaen, F.: Coherent risk measures. Monograph, Scuola Normale, Superiore, Pisa (2002).

\bibitem{Denneberg1990} Denneberg, D.:  Premium calculation: why standard deviation should be replaced by absolute deviation. ASTIN Bulletin {\bf 20}, 181-190 (1990)


\bibitem{EdgarSucheston1992} Edgar, G.A., Sucheston, L.: Stopping times and directed processes. Cambridge University Press, Cambridge (1992).

\bibitem{ElKarouiRavanelli} El Karoui, N., Ravanelli, C.: Cash subadditive risk measures and interest rate ambiguity. Mathematical Finance {\bf 19}, 561--590 (2009)

\bibitem{Farkasetal}  Farkas, W.,  Koch-Medina, P., Munari, C.-A.: Beyond cash-additive capital requirements: when changing the numeraire fails. 	 arXiv:1206.0478 (2012)

\bibitem{FilipovicSvindland} Filipovic, D., Svindland, G.:  The canonical model space for law-invariant convex risk measures is $L^1$. Mathematical Finance {\bf 22}, 585--589 (2012)

\bibitem{FoellmerSchied2002} F\"ollmer, H., Schied, A.: Convex measures of risk and trading constraints. Finance and Stochastics {\bf 6}, 429--447 (2002) 

\bibitem{FoellmerSchied2011} F\"ollmer, H., Schied, A.: Stochastic finance. An introduction in discrete time, de Gruyter, Berlin (2011) (3rd ed.).

\bibitem{GilatHelmers1997} Gilat, D., Helmers, R.: On strong laws for generalized L-statistics with dependent data. Commentationes Mathtematicae Universitatis Carolinae {\bf 38}, 187--192 (1997)

\bibitem{Hampel} Hampel,  F.R.: A general qualitative definition of robustness. Annals of Mathematical Statistics {\bf 42}, 1887--1896 (1971)

\bibitem{Huber}  Huber, P.J., Ronchetti, E.M.: Robust Statistics, Wiley, New York, (2009) (2nd ed).


\bibitem{Kallenberg1997} Kallenberg, O.: Foundations of modern probability. Springer, New York (1997).

\bibitem{Kraetschmer et al 2012} Kr\"atschmer, V., Schied, A., Z\"ahle, H.: Qualitative and infinitesimal robustness of tail-dependent statistical functionals.  Journal of Multivariate Analysis, {\bf 103}, 35--47 (2012)

\bibitem{Kusuoka} Kusuoka, S.: On law invariant coherent risk measures. Adv.
Math. Econ. {\bf 3}, 83--95 (2001).


\bibitem{Mizera2010} Mizera, I. (2010) Qualitative robustness and weak continuity: the extreme unction. In: Nonparametrics and robustness in modern statistical inference and time series analysis: a Festschrift in honor of Professor Jana Jure?ková. IMS Collections Festschrift, Institute of Mathematical Statistics, Beachwood, OH, pp.169--181 (2010).

\bibitem{Nelson} Nelson, D.B.: Stationarity and persistence in the GARCH(1,1) model. Econometric Theory {\bf 6}, 318--334 (1990).

\bibitem{Parthasarathy} Parthasarathy, K. R.: Probability measures on metric spaces (Probability and Mathematical Statistics, No.~3), Academic Press, Inc., New York-London (1967).



\bibitem{RuSha} Ruszczynski, A., Shapiro, A.: Optimization of convex risk functions.  Mathematics of Operations Research {\bf 31}, 433--451 (2006).

\bibitem{Schied2004} Schied, A.: On the Neyman--Pearson problem for law-invariant risk measures and robust utility functionals.  Annals of Applied Probability {\bf 14}, 1398--1423 (2004). 

\bibitem{Tsukahara}Tsukahara, H.: Estimation of distortion risk measures. Forthcoming in  the Journal of Financial Econometrics.

\bibitem{van der Vaart Wellner 1996} van der Vaart, A.W., Wellner, J.A.: Weak convergence and empirical processes. Springer, New York (1996).


\bibitem{vanZwet}van Zwet, W.R.: A strong law for linear functionals of order statistics. The Annals of Probability {\bf 8}, 986--990 (1980). 


\bibitem{Villani2003} Villani, C.: Topics in optimal transportation. American Mathematical Society, Providence, RI (2004).


\bibitem{Wang1996} Wang, S.S.: Premium calculation by transforming the layer premium density. ASTIN Bulletin {\bf 26}, 71--92 (1996)


\bibitem{Weber} Weber, S.: Distribution-invariant risk measures, information, and dynamic consistency. Mathematical Finance {\bf 16}, 419--441 (2006). 

\bibitem{Yaari} Yaari, M.: The dual theory of choice under risk. Econometrica {\bf 55}, 95--115 (1987)

\bibitem{Zaehle2013} Z\"ahle, H.:  Marcinkiewicz--Zygmund and ordinary strong laws for empirical distribution functions and plug-in estimators. Forthcoming in Statistics (DOI:10.1080/02331888.2013.800075).
\end{thebibliography}
\end{document}